\documentclass[11pt]{article}

\usepackage{csquotes}
\usepackage[english]{babel}
\usepackage[utf8]{inputenc}
\usepackage[T1]{fontenc}
\usepackage{lmodern}
\usepackage{fullpage}
\usepackage{graphicx}
\usepackage{todonotes}
\usepackage{color}
\definecolor{darkgreen}{rgb}{0,0.7,0}
\usepackage{hyperref}
\hypersetup{
    unicode=false,          
    colorlinks=true,        
    linkcolor=red,          
    citecolor=darkgreen,        
    filecolor=magenta,      
    urlcolor=cyan           
}

\usepackage{amsthm}
\usepackage{amsmath}
\usepackage{amssymb}
\usepackage{amsfonts}
\usepackage{mathrsfs}
\usepackage{mathtools}
\usepackage{verbatim}
\usepackage{footnote}
\usepackage[linesnumbered,vlined,ruled]{algorithm2e}
  \DontPrintSemicolon
  \SetKwInOut{Input}{Input}
  \SetKwInOut{Output}{Output}
  
  \SetCommentSty{mycommfont}

\usepackage{lineno}
\usepackage{caption}
\usepackage{framed}
\usepackage[framemethod=tikz]{mdframed}
\global\mdfdefinestyle{myframe}{leftmargin=.75in,rightmargin=.75in,linecolor=black,linewidth=1.5pt,innertopmargin=10pt,innerbottommargin=10pt} 
\usepackage{enumerate}

\usepackage{tikzsymbols}
\usepackage{thmtools,thm-restate}
\usepackage{nicefrac}

\usepackage[capitalize, nameinlink]{cleveref}
\usepackage[section]{placeins}
\crefname{theorem}{Theorem}{Theorems}
\Crefname{lemma}{Lemma}{Lemmas}
\Crefname{invariant}{Invariant}{Invariants}
\Crefname{claim}{Claim}{Claims}
\Crefname{observation}{Observation}{Observations}
\Crefname{algorithm}{Algorithm}{Algorithms}

\DeclareMathOperator{\poly}{poly}

\newcommand{\eqdef}{\stackrel{\text{\tiny\rm def}}{=}}

\newtheorem{theorem}{Theorem}
\newtheorem{lemma}[theorem]{Lemma}

\newtheorem{definition}[theorem]{Definition}

\newtheorem{invariant}[theorem]{Invariant}

\newcommand{\rb}[1]{\left( #1 \right)}
\newcommand{\ee}[1]{\mathbb{E} \left[ #1 \right]}
\newcommand{\prob}[1]{\Pr \left[ #1 \right]}

\newcommand{\cE}{\mathcal{E}}
\newcommand{\cO}{\mathcal{O}}
\newcommand{\cS}{\mathcal{S}}
\newcommand{\cC}{\mathcal{C}}

\newcommand{\hcE}{\hat{\mathcal{E}}}
\newcommand{\hcS}{\hat{\mathcal{S}}}
\newcommand{\hd}{\hat{d}}

\newcommand{\Oracle}{\cO^{\mathsf{set}}}
\newcommand{\OracleEl}{\cO^{\mathsf{el}}}
\newcommand{\opt}{\mathrm{OPT}}
\newcommand{\OPT}{\opt}
\newcommand{\Bin}{\mathrm{Bin}}
\newcommand{\sN}{\mathcal{N}}

\newcommand{\RecSplit}{\textsc{RecursiveSplitting}}

\newcommand{\sml}{\mathsf{small}}
\newcommand{\bg}{\mathsf{big}}
\newcommand{\cvr}{\mathsf{cover}}

\newcommand{\nrml}{\mathsf{normal}}
\newcommand{\logfive}{\log^{10}s \cdot \log^{10} t}
\newcommand{\logten}{\log^{20} s \cdot \log^{20} t}
\newcommand{\negfive}{O\rb{\frac{1}{(st)^{\log^5 s \cdot \log^5 t}}}}
\newcommand{\negfour}{O\rb{\frac{1}{(st)^{\log^4 s \cdot \log^4 t}}}}
\newcommand{\negthree}{O\rb{\frac{1}{(st)^{\log^3 s \cdot \log^3 t}}}}

\newcommand{\negonenoO}{\frac{1}{(st)^{\log s \cdot \log t}}}
\newcommand{\setsuperi}{\cS_{\bg,i}}
\newcommand{\setsnormal}{\cS_{\nrml,i}}
\newcommand{\setssmall}{\cS_{\sml}}
\newcommand{\xei}{X_e^{(i)}}
\newcommand{\degree}{\hat{d}_{B_i}}
\newcommand{\stlogslogt}{(st)^{O(\log{s} \cdot \log{t})}}

\title{Improved Local Computation Algorithm for Set Cover via Sparsification}
\author{Christoph Grunau%
			\thanks{ETH Zurich, \protect\url{cgrunau@student.ethz.ch}}
			\and
			Slobodan Mitrovi\'{c}%
			\thanks{CSAIL, MIT; \protect\url{slobo@mit.edu}}
			\and
			Ronitt Rubinfeld%
			\thanks{CSAIL, MIT and TAU; \protect\url{ronitt@csail.mit.edu}}
			\and
			Ali Vakilian%
			\thanks{University of Wisconsin - Madison; \protect\url{vakilian@wisc.edu}. 
			This work was done when the author was a graduate student at MIT.}}
\date{}

\begin{document}

\maketitle

\begin{abstract}
	We design a Local Computation Algorithm (LCA) for the set cover problem. Given a set system where each set has size at most $s$ and each element is contained in at most $t$ sets, the algorithm reports whether a given set is in some fixed set cover whose expected size is $O(\log{s})$ times the minimum fractional set cover value. Our algorithm requires $s^{O(\log{s})} t^{O(\log{s} \cdot (\log \log{s} + \log \log{t}))}$ queries. This result improves upon the application of the reduction of [Parnas and Ron, TCS'07] on the result of [Kuhn et al., SODA'06], which leads to a query complexity of $(st)^{O(\log{s} \cdot \log{t})}$.

	To obtain this result, we design a parallel set cover algorithm that admits an efficient simulation in the LCA model by using a {\em sparsification} technique introduced in [Ghaffari and Uitto, SODA'19] for the maximal independent set problem. The parallel algorithm adds a random subset of the sets to the solution in a style similar to the PRAM algorithm of [Berger et al., FOCS'89]. However, our algorithm differs in the way that it never \emph{revokes} its decisions, which results in a fewer number of adaptive rounds. This requires a novel approximation analysis which might be of independent interest. 
\end{abstract}

\pagenumbering{gobble}

\newpage
\setcounter{page}{1}
\pagenumbering{arabic}

\section{Introduction}
The set cover problem is one of the classical problems in optimization and computer science. In this problem, we are given a universe of $n$ elements $\cE$ and a family of $m$ sets $\cS \subseteq 2^{\cE}$, and our goal is to find a minimum size {\em set cover} of $\cE$; i.e., a collection of sets in $\cS$ whose union is equal to $\cE$. 
The set cover problem is a well-studied problem with applications in many areas such as machine learning, data mining and operation research~\cite{gw-ceaas-97,sg-mcsma-09, kv-iclt-1994, bansal2009new}.  

A simple greedy algorithm for this problem, that repeatedly adds a set containing the largest number of yet uncovered elements to the cover, guarantees a $O(\log n)$-approximation~\cite{johnson1974approximation,lovasz1975ratio}. Unless $P = NP$, this approximation guarantee is within a constant factor compared  to the approximation guarantee of any polynomial-time algorithm~\cite{rs-sbepl-97, feige1998threshold, ams-acskr-06, m-pgcnsc-12, ds-aapr-14}. Unfortunately, this standard greedy algorithm does not scale very well for massive data sets (e.g., see Cormode~et al.~\cite{ckw-scavl-10} for an experimental evaluation of the greedy algorithm on large data sets). This difficulty has led to considerable interest in designing set cover algorithms for computational models tailored to process massive amounts of data such as parallel computation~\cite{berger1994efficient, blelloch2011linear, blelloch2012parallel}, streaming~\cite{sg-mcsma-09, kmvv-fgams-13, er-sssc-14, dimv-sccsc-14, cw-igpcs-16, himv-ttbsscp-16, aky-tbspscscp-16, a-tsatmpsscp-17, bem-aosacp-17, imruvy-fscsm-17}, sublinear time/query algorithms~\cite{grigoriadis1995sublinear, Koufogiannakis2014, IndykMRVY18} and local computation algorithms~\cite{nguyen2008constant, yoshida2012improved}. 


In many scenarios we are interested in designing {\em extremely fast} algorithms for learning only a minuscule portion of a solution, rather than computing and storing the entire one.
This has led to the model of {\em local computation algorithms (LCAs)} introduced by Rubinfeld et al.~\cite{rubinfeld2011fast} and Alon et al.~\cite{alon2012space}. 
A LCA for the set cover problem provides oracle access to some fixed set cover $\cC$.\footnote{Note that for any given 
input, there may be many set covers. $\cC$ is a unique set cover that depends only on the input and the random bits 
used by the algorithm.} That is, given some arbitrary set $S$, the LCA needs to report whether $S$ is part of the set cover $\cC$ by only having primitive query access to the input set system $(\cS,\cE)$. The only shared state across different oracle calls is a tape of random bits. The performance is measured by the approximation guarantee and the query complexity of the LCA. In the LCA model, we further assume that each set in $\cS$ has size at most $s$ and each element in $\cE$ is contained in at most $t$ sets. \cref{section:LCA} contains additional information about the underlying computational model.

\paragraph{Simulating the Greedy Algorithm}
There are two main approaches for designing LCAs for set cover. The first one is based on simulating the greedy algorithm using a randomized ranking technique due to Nguyen and Onak~\cite{nguyen2008constant} which was later improved by Yoshida~et al.~\cite{yoshida2012improved}. While the goal of~\cite{nguyen2008constant,yoshida2012improved} is the design of constant-time algorithms for approximating the size of an optimal set cover solution on bounded degree instances, one can turn their algorithms into $O(\log s)$-approximation LCAs with a respective query complexity of  $2^{O((st)^4\cdot 2^s)}$ ~\cite{nguyen2008constant} and $(st)^{O(s)}$ ~\cite{yoshida2012improved}. 

\paragraph{LCAs via Distributed Algorithms}
The other approach is via a generic reduction from distributed algorithms to LCAs by Parnas and Ron~\cite{parnas2007approximating}. Kuhn et al.~\cite{kuhn2006price} designed a distributed algorithm that computes a $O(1)$-approximate fractional solution in $O(\log{s} \cdot \log{t})$ many rounds. In each round, each set and each element only requires local information, i.e., a set uses only the information of the elements it contains and an element uses only the information of the sets that it is contained in. By the reduction of~\cite{parnas2007approximating}, the distributed algorithm of \cite{kuhn2006price} can be transformed into an LCA with a query complexity of $(s t)^{O(\log{s} \cdot \log{t})}$. The LCA only outputs the fractional contribution of $S$ in the corresponding fractional set cover solution for some given set $S$. However, via  the standard randomized rounding technique and a slight increase in the query complexity by a factor of $ s \cdot t$, the LCA can be turned into an LCA that outputs an $O(\log s)$-approximate {\em integral} solution for the set cover problem. 
There is no known distributed algorithm for set cover that uses fewer than $O(\log{s} \cdot \log{t})$ iterations of computation, and hence applying known reductions to any such algorithm would not improve the query complexity implied by \cite{kuhn2006price} and \cite{parnas2007approximating}. 
It is thus natural to wonder:

\begin{displayquote}
\emph{Is it possible to obtain LCAs for the set cover problem with a smaller query complexity \\ compared to the query complexity obtained via standard reductions to distributed algorithms?}
\end{displayquote}
In this paper, we answer this question affirmatively by presenting the following result.
\begin{theorem}\label{thm:main-LCA}
	There exists an LCA with a query complexity of $s^{O(\log{s})} t^{O(\log{s} \cdot (\log \log{s} + \log \log{t}))}$ that produces a set cover with an expected size of $O(\log s) \cdot \OPT$, where $\OPT$ denotes the value of an optimal fractional set cover solution.
\end{theorem}
This result is proved in \cref{sec:log-log-t}. In \cref{sec:our-techniques} we give an overview of our approach.

\subsection{The Local Computation Model} \label{section:LCA}
We consider a graph representation of the input set system $(\cS, \cE)$ and adopt the definition of LCAs by Rubinfeld et al.~\cite{rubinfeld2011fast}. In the design of LCAs, the query access model to the input instance plays an important role. Here we consider standard {\em neighbor} queries: ``what is the $i$-th element in a given set $S$?'' and ``what is the $i$-th set containing a given element $e$?''. In our analysis, we assume that a query to a set (resp., element) returns all elements contained in (resp., all sets containing) the required set (resp., element). This will increase the query complexity only by a factor of $O(s+t)$.
\begin{definition}[LCA for set cover]
An LCA for the set cover problem is a (randomized) algorithm that has access to neighbor queries, a tape of random bits and local working memory. Given some set $S$ from the input set system $(\cS, \cE)$, the LCA algorithm returns whether $S$ is part of the set cover by making queries to the input. The answer must only depend on the given set $S$, the input set system $(\cS, \cE)$ and the random tape. Moreover, for a fixed tape of random bits, the answer given by the LCA to all sets, must be consistent with one particular valid set cover.

We remark that by running the described LCA on at most $t$ sets, we can determine which set is covering a given element in the approximate set cover constructed (hypothetically) by the LCA. 
\end{definition}
The main complexity measures of the LCA for set cover are the expected size of the solution, as well as the query complexity of the LCA which is defined as the maximum number of queries that the algorithm makes to return an answer for any arbitrary element $e$ in the input universe. 
Note that all the algorithms in this paper are randomized and, for any input set system, provide guarantees on the expected size of the solution, over the random tape. For simplicity, we describe all our randomized algorithms using full independence;
however, they can be implemented using a seed with a polylogarithmic number of random bits by the techniques and concentration bounds in~\cite{schmidt1995chernoff,vadhan2012pseudorandomness}.

\subsection{Roadmap and Our Technical Contributions}
\label{sec:our-techniques}
\paragraph{The base algorithm (\cref{sec:base-algorithm}).}
The starting point of our approach is a parallel algorithm that constructs a set cover in  $\log{s} \cdot \log{t}$ rounds. This algorithm, that we present in \cref{sec:base-algorithm}, consists of $\log{s}$ \emph{stages} and each stage consists of $\log{t}$ \emph{iterations}. The stages are used to process sets that have a large number of uncovered elements. That is, in stage $i$ we consider all the sets that have at least $s / 2^i$ uncovered elements. We call such sets \emph{large}.

After the execution of stage $i$, a large set is either added to the set cover constructed so far, or the number of uncovered elements it contains dropped below $s / 2^i$ during the execution of stage $i$. Iterations are used to progressively add large sets to the cover, while attempting to assure that an element is not covered by too many large sets. (We make this statement formal in \cref{Lemma:constant coverage}.) More precisely, in the $k$-th iteration, each currently large set is added with probability $2^k / t$. This step is applied for all the large sets simultaneously. Again, adding large sets with a probability of $2^k/t$ is supposed to ensure that an element contained in roughly $t / 2^k$ large sets will be covered $\Theta(1)$ times at the moment it gets covered for the first time. We show that having this kind of guarantee suffices to obtain an expected $O(\log{s})$-approximate minimum set cover.

The base algorithm has a parallel depth of $O(\log{s} \cdot \log{t})$ and directly simulating it in the LCA model would result in a query complexity of $\stlogslogt$. In spirit, our parallel algorithm is similar to the algorithm developed in \cite{berger1994efficient}, but there is also a crucial difference. After randomly picking a family of sets in each round, the algorithm in \cite{berger1994efficient} verifies whether the number of newly covered elements is large enough in comparison to the number of chosen sets. If yes, then they add the family of sets to the set cover. If not, they repeat the selection process. A random family of sets is good with a constant probability, yielding an O(log n) bound on the number of times that the algorithm needs to repeat a round. Therefore, the resulting algorithm has a larger parallel depth compared to our algorithm.

\paragraph{Estimating set sizes (\cref{sec:generic-algorithm}).}

In \cref{sec:generic-algorithm} we design and analyze \cref{alg:generic-alg}. Compared to \cref{alg:base}, \cref{alg:generic-alg} only estimates the number of remaining free elements of a set. This is a crucial step towards getting query-efficient LCAs. However, for reasons explained in \cref{sec:generic-algorithm}, \cref{alg:generic-alg} does not admit an efficient LCA simulation. The reason for including \cref{alg:generic-alg} is twofold. First, all the later algorithms that admit an efficient LCA simulation, are basically identical to \cref{alg:generic-alg}, except that they ensure that certain conditions hold. Second, the approximation guarantee of  \cref{alg:generic-alg} can be established in more or less the same way as the approximation guarantee for \cref{alg:base}. 
In contrast, establishing the approximation guarantee for \cref{alg:sqrt-t} and \cref{alg:main-alg} is much harder. Therefore, it is easier to establish the approximation guarantee of \cref{alg:sqrt-t} and \cref{alg:main-alg} by relating them to \cref{alg:generic-alg}.

\paragraph{Sparsification of the element-neighborhoods (\cref{sec:sqrt-log-t,sec:log-log-t}).}
Estimating the number of uncovered elements within a given set, instead of counting them exactly, reduces the number of queries that a set has to perform. However, this optimization in terms of set-size estimates does not affect the number of direct set-queries that an element has to perform. To achieve our advertised query complexity, we also reduce the number of queries directly performed by an element. Our high-level approach for reducing the number of queries for elements follows the lines of the work \cite{ghaffari2019sparsifying,chang2019complexity} (and in \cref{sec:sqrt-log-t} also the work \cite{onak2018round}). \cite{ghaffari2019sparsifying} design an LCA for maximal independent set (MIS) and maximal matching. These LCAs do not have a generic reduction to other models.

One of the main challenges in adapting the approach of \cite{ghaffari2019sparsifying} to the set cover problem is that, in the case of MIS or maximal matching it is needed to handle only one type of objects, i.e., vertices or edges. However, in the case of set cover, our algorithm handles sets and elements simultaneously.




\section{Preliminaries}
\paragraph{Notation.}
We will use $\cS$ (respectively, $\cE$) with added sub- and super-scripts to refer to subsets of $\cS$ (respectively, $\cE$). Given a set $S$ at some step of an algorithm, we use $d(S)$ to denote the number of uncovered elements of $S$ in the current set cover. When we use $\hd(S)$, it refers to an estimated number of uncovered elements of $S$ in the current set cover.

Uncovered elements are also referred to by \emph{free}. Subscripts $i$, $j$, and $k$ will have the following meaning: $i$ refers to a stage, $j$ refers to a phase, and $k$ refers to an iteration.

\paragraph{Relevant Concentration Bounds.} Throughout the paper, we will use the following well-known variants of Chernoff bound.
\begin{theorem}[Chernoff bound]\label{lemma:chernoff}
	Let $X_1, \ldots, X_k$ be independent random variables taking values in $[0, 1]$. Let $X \eqdef \sum_{i = 1}^k X_i$ and $\mu \eqdef \ee{X}$. Then,
	\begin{enumerate}[(A)]
		\item\label{item:delta-at-most-1-le} For any $\delta \in [0, 1]$ it holds $\prob{X \le (1 - \delta) \mu} \le \exp\rb{- \delta^2 \mu / 2}$.
		\item\label{item:delta-at-least-1} For any $\delta \ge 1$ it holds $\prob{X \ge (1 + \delta) \mu} \le \exp\rb{- \delta \mu / 3}$.
	\end{enumerate}
\end{theorem}

\section{The Base Algorithm}
\label{sec:base-algorithm}
We now fully state and analyze the base algorithm briefly described in \cref{sec:our-techniques}. This algorithm is presented as \cref{alg:base}. As a reminder, this algorithm constructs a set cover in $\log (s)$ many stages (\cref{line:base-stages} of \cref{alg:base}) and each stage consists of $\log (t)$ many iterations (\cref{line:base-iterations} of \cref{alg:base}). \cref{alg:base} maintains the invariant that no set contains more than $s / 2^i$ uncovered elements at the end of stage $i$. In iteration $k$ of stage $i$, each set containing at least $s / 2^i$ uncovered elements is added to the set cover with a probability of $2^k / t$. This step is executed for all the sets simultaneously. In \cref{sec:base-correctness} we show that the algorithm indeed outputs a set cover, while in \cref{section:approx-base-algorithm} we analyze the approximation guarantee. Our approximation analysis diverts from the prior work known to the authors and might be of independent interest to the reader.
\begin{algorithm}	
	\For{stage $i = 1$ to $\log{s}$ \label{line:base-stages}}{
		\For{iteration $k = 1$ to  $\log{t}$ \label{line:base-iterations}}{
			\For{each sets $S$ in parallel}{
				\If(\tcp*[h]{$d(S)$ denotes the number of free (or yet uncovered) elements in $S$}){$d(S) \geq s/2^i$ \label{line:base-large-set}}{ 
					{\bf add} $S$ to $\cS_{\cvr}$ with probability $2^k/t$. \label{Line:Alg1 add sets}
				}
			}				
		}
	}
	\Return{$\cS_{\cvr}$}
	\caption{The base algorithm for our LCAs which runs in $O(\log{s} \cdot \log{t})$ iterations.}\label{alg:base}
\end{algorithm}

\subsection{Correctness}
\label{sec:base-correctness}
It is not hard to see that \cref{alg:base} indeed constructs a set cover. This stems from the fact that in iteration $k = \log{t}$, each set having at least $s / 2^i$ uncovered elements at that point is added to the cover. A more elaborate argument is provided in the proof of the following claim.
\begin{lemma}
	For $s \geq 2$ and $t \geq 2$, \cref{alg:base} returns a valid set cover.
\end{lemma}
\begin{proof}
	It suffices to show that each element is covered by some set in the solution. Consider an arbitrary element $e$. If $e$ was covered before the last iteration of the last stage, then we are fine. Otherwise, there is at least one set S which contains $e$ and was not chosen before the last iteration of the last stage. As this set has at least one free element, it is added to the set cover with a probability of $2^{\log t}/t = 1$. Hence, $e$ will be covered by $S$.
\end{proof}

\subsection{Analysis of the Approximation Guarantee}\label{section:approx-base-algorithm}
In this section we analyze the approximation guarantee of \cref{alg:base}. Our analysis shows that the algorithm matches, up to a constant factor and in expectation, the guarantee of the best possible polynomial-time algorithm unless $\mathrm{P} = \mathrm{NP}$.
\begin{theorem} \label{Theorem:approx_basic}
	Let $(\cS,\cE)$ be some set cover instance and let $\cS_{\cvr}$ denote the solution returned by~\cref{alg:base}. Furthermore, let $\opt$ denote the value of an optimal fractional set cover solution for $(\cS,\cE)$. Then, $\ee{|\cS_{\cvr}|} = O(\log s)\cdot \opt$.
\end{theorem}
To prove this claim, we need the following result that we establish in the rest of this section.
\begin{lemma} \label{Lemma:first_iteration}
	Let $(\cS,\cE)$ be a set cover instance with a maximal set size of $s$ and let $\cS_{\cvr,1}$ denote the number of sets that are added to $\cS_{\cvr}$ during the first stage of \cref{alg:base}. Then, $\mathbb{E}[|\cS_{\cvr,1}|] = O(n/s)$.
\end{lemma}
\begin{proof}[Proof of \cref{Theorem:approx_basic}]
In \cref{Lemma:first_iteration}, we show that \cref{alg:base} adds at most $O(n/s)$ many sets, in expectation, to the set cover during the first stage. Here we show that this observation suffices to prove the theorem. To see why, let $(\cS_i,\cE_i)$  denote the set cover instance with $\cE_i$ being the set of all uncovered elements prior to the $i$-th stage and $\cS_i = \{S \cap \cE_i|S \in \cS\}$. Let $n_i = |\cE_i|$ denote the number of uncovered elements prior to the $i$-th stage. Note that the maximal set size in $S_i$ is at most $s_i := s/2^{i-1}$. Let $\opt_i$ denote the value of an optimal fractional set cover for $(\cS_i,\cE_i)$. 
\begin{align}\label{eq:incremental-sol}
n_i/s_i \leq \opt_i \leq \opt
\end{align}

Let $\cS_{\cvr,i}$ denote the collection of sets that \cref{alg:base} picks in stage $i$. Note that $\ee{|\cS_{\cvr,i}|}$ is equal to the expected number of sets that \cref{alg:base} would pick in the first stage when given $(\cS_i,\cE_i)$ as an input. Thus, by~\cref{eq:incremental-sol}, $\ee{|\cS_{\cvr,i}|} = O(n_i/s_i) = O(\opt)$ which implies that 
\begin{align*}
	\ee{|\cS_{\cvr}|} = \sum_i \ee{|\cS_{\cvr,i}|} \leq \log(s) \cdot O(\opt) = O(\log s) \cdot \opt.
\end{align*}
\end{proof}
The following is the main technical lemma that we use to prove~\cref{Lemma:first_iteration}.
\begin{lemma}\label{Lemma:constant coverage}
Let $e \in \cE$ be an arbitrary element and $X_e$ be a random variable which is equal to $0$ if $e$ does not get covered during the first stage and otherwise is equal to the number of sets that contain $e$ and are added to $\cS_{\cvr}$ in the same iteration in which $e$ is covered for the first time. Then, $\mathbb{E}[X_e] \leq 5$.   
\end{lemma}
Let's observe what happens with $e$ during the first stage. In each iteration, a set $S$ is {\em active} if it contains $e$ and has more than $s/2$ free elements in the beginning of the iteration.
Suppose that there are $N_1\leq t$ active sets in the first iteration.  Since each active set joins the set cover independently with probability $2/t$, the number of times that $e$ is covered in the first iteration can be described by a random variable $Y_1 \sim \Bin(N_1,2/t)$. 
If $Y_1 > 0$, then $X_e = Y_1$. 
Otherwise, $e$ is still uncovered after the first iteration. In this case, there are $N_2$ active sets at the beginning of the second iteration. 
Note that as the number of free elements in some of the $N_1$ active sets of the first iteration might have dropped below $s/2$, $N_1\geq N_2$. 
More generally, given that $e$ was not covered prior to the $k$-th iteration, there are $N_k$ active sets and the number of times that $e$ is covered in this iteration can be described as a random variable $Y_k \sim \Bin(N_k,2^k/t)$. If $Y_k> 0$, then $X_e = Y_k$, otherwise we proceed to the next iteration. Note that for all $k\geq1$, $N_k \geq N_{k+1}$ and $N_{k+1}$ are random variables.
To analyze the random process it is however easier to think of $N_1, \cdots, N_{\log t}$ as being fixed in advance. 
To that end, we think about an equivalent random process. We first choose a random vector $b_S \in \{0,1\}^{\log(t)}$ for each set $S$ such that $\Pr[(b_S)_k = 1] = 2^k/t$. 
Now, in the $k$-th iteration we add a set $S$ if it is an active set in the $k$-th iteration and $(b_S)_k = 1$. It is easy to see that this process is equivalent. Additionally, it allows us to use the principle of deferred decision making by fixing all the random vectors for those sets which do not contain $e$ in advance. After fixing those, it is easy to see that $N_k$ is also fixed if $e$ is a free element in the $k$-th iteration. 
\begin{definition}
Let $seq = (n_1, ..., n_{\log t})$ be an integer sequence such that $t \geq n_1 \geq n_2 \geq ... \geq n_{\log t} \geq 0$. Let $Y_1, ..., Y_{\log t}$ be random variables with $Y_k \sim \Bin(n_k,2^k/t)$. Let $Y_{seq} = Y_{k^*}$ with $k^*$ being the smallest index for which $Y_{k^*} > 0$ if such a $k^*$ exists and $0$ otherwise. 
\end{definition}
We now prove the following result that we will use to upper-bound	 $\ee{X_e}$.
\begin{lemma}\label{lemma:upper-bound-ee-Xe-by-max}
	$\ee{X_e}$ can be upper-bounded as follows
	\[
		\ee{X_e} \le \max_{t\geq n_1 \geq \ldots \geq n_{\log t} \geq 0} \ee{Y_{(n_1, \ldots, n_{\log t})}}.
	\]
\end{lemma}
\begin{proof}
	Let $\sN$ denote the collection of all possible values for the random variable ($N_1,\ldots, N_{\log t})$.
	The proof is given by the following sequence of inequalities:
	\begin{align*}
		\ee{X_e} &= \sum_{(n_1,\ldots,n_{\log t}) \in \sN}\ee{X_e|(N_1,\ldots,N_{\log t}) = (n_1, \ldots, n_{\log t})} \cdot \Pr[(N_1,\ldots,N_{\log t}) = (n_1,\ldots, n_{\log t})] \\
		&= \sum_{(n_1,\ldots,n_{\log t}) \in \sN}\ee{Y_{(n_1, \ldots, n_{\log t})}} \cdot \Pr[(N_1, \ldots,N_{\log t}) = (n_1,\ldots, n_{\log t})] \\
		& \leq \max_{(n_1,\ldots,n_{\log t}) \in \sN} \ee{Y_{(n_1, \ldots, n_{\log t})}} \\
		&\leq \max_{t\geq n_1 \geq \ldots \geq n_{\log t} \geq 0} \ee{Y_{(n_1, \ldots, n_{\log t})}}
	\end{align*}
\end{proof}
Following \cref{lemma:upper-bound-ee-Xe-by-max}, it only remains to upper-bound $\ee{Y_{(n_1, \ldots, n_{\log t})}}$ for any arbitrary sequence $n_1, \cdots, n_{\log t}$ such that $t \geq n_1 \geq \ldots \geq n_{\log t} \geq 0$. We will do this in two steps.

First, we will find an upper bound for $\ee{Y_{(n_1, \ldots, n_{\log t})}}$ and define a function $f$ such that $f((n_1, \ldots, n_{\log t}),t)$ is equal to this upper bound. Then, we will show by induction on the sequence length that for some restrictions on the input parameters the value of the function can be bounded by a constant. The restricted input parameters still capture the parameters produced by \cref{alg:base} and what we require for the analysis. We upper-bound $\ee{Y_{(n_1, \ldots, n_{\log t})}}$ as follows
\begin{align*}
	\ee{Y_{(n_1, \ldots, n_{\log t})}} &= \ee{Y_1} + \sum_{k=2}^{\log t} \ee{Y_k} \cdot \Pr[Y_1 = 0, Y_2 = 0, \ldots, Y_{k-1} = 0] \\
	&= n_1 \cdot \frac{2}{t} + \sum_{k=2}^{\log t} n_k \cdot \frac{2^k}{t}\prod_{j=1}^{k-1}\rb{1-2^j/t}^{n_j}\\&\leq n_1 \cdot \frac{2}{t} + \sum_{k=2}^{\log t} n_k \cdot \frac{2^k}{t}\prod_{j=1}^{k-1}e^{-n_j\cdot\frac{2^j}{t}}.
\end{align*} 

\begin{definition} \label{def:seq}
	Let $f((x_1, \ldots, x_l),r) := x_1 \cdot \frac{2}{r} + \sum_{k=2}^{l} x_k \cdot \frac{2^k}{r} \cdot \prod_{j=1}^{k-1}e^{-x_j \cdot \frac{2^j}{r}}$. For $l \in \mathbb{N}$, let $P(l)$ be the property that for all $(x_1, \ldots, x_l) \in \mathbb{R}_{>0}$ and $r \in \mathbb{R}_{>0}$ such that $r \geq x_1 \geq x_2 \geq \ldots \geq x_l \geq 0$, we have $f((x_1, \ldots, x_l),r) \leq 5$.
\end{definition}
We have that $P(\log(t))$ being true directly implies $\ee{Y_{(n_1,\ldots,n_{\log t})}} \leq 5$, which in turn implies \cref{Lemma:constant coverage}.
\begin{lemma} \label{lemma:seq}
Let $P(l)$ be as defined in \cref{def:seq}. Then, $P(l)$ holds for all $l \geq 1$.
\end{lemma}
\begin{proof}
	We prove this statement by induction on $l$. We first show $P(1)$. Let $x_1$ and $r \in \mathbb{R}_{>0}$ be arbitrary such that $r \geq x_1 \geq 0$. We have $f((x_1),r) = x_1 \cdot 2/r \leq  5$.
	
	Now, let $l \in \mathbb{N}$ be arbitrary. We assume $P(l)$ to be true and need to show $P(l+1)$. Let $r \in \mathbb{R}_{>0}$ and $x_1, \ldots, x_{l+1}$ be arbitrary such that $r \geq x_1 \geq \ldots \geq x_{l+1} \geq 0$.
	
	\begin{align*}
		f((x_1, \ldots, x_{l+1}),r) &= x_1 \cdot \frac{2}{r} + \sum_{k=2}^{l+1} x_k \cdot \frac{2^k}{r} \cdot \prod_{j=1}^{k-1}e^{-x_j \cdot \frac{2^j}{r}} \\
		&= x_1 \cdot \frac{2}{r} + e^{-x_1\cdot \frac{2}{r}} \cdot \rb{x_2 \cdot \frac{2^2}{r} + \sum_{k=3}^{l+1} x_k \cdot \frac{2^k}{r} \cdot \prod_{j=2}^{k-1}e^{-x_j \cdot \frac{2^j}{r}}} \\
		&= x_1 \cdot \frac{2}{r} + e^{-x_1\cdot \frac{2}{r}} \cdot \rb{(2x_{2}) \cdot \frac{2}{r} + \sum_{k=2}^{l} (2x_{k+1}) \cdot \frac{2^k}{r} \cdot \prod_{j=1}^{k-1}e^{-(2x_{j+1}) \cdot \frac{2^j}{r}}} \\
		&=  x_1 \cdot \frac{2}{r} + e^{-x_1\cdot \frac{2}{r}} \cdot f((2x_2,2x_3, \ldots ,2x_{l+1}),r)
	\end{align*}
	We would like to use the induction hypothesis to bound $f((2x_2,2x_3, \ldots,2x_{l+1}),r)$. However, it might be the case that $2x_2 > r$. Therefore we need to do a case distinction. First, assume that $x_1 \leq r/2$ which implies that $2x_2 \leq r$. Thus we can indeed use the induction hypothesis:
	
	\begin{align*}
		f((x_1,\ldots,x_{l+1}),r) &\leq x_1 \cdot \frac{2}{r} + e^{-x_1 \cdot \frac{2}{r}} \cdot 5 \leq 5 
	\end{align*}
	
	Next we consider the case where $2x_2 > r$. 
	
	\begin{align*}
			f((x_1, \ldots, x_{l+1}),r) &\leq x_1 \cdot \frac{2}{r} + e^{-x_1\cdot \frac{2}{r}} \cdot f((2x_2,2x_3,\ldots,2x_{l+1}),r) \\
			&\leq x_1 \cdot \frac{2}{r} + e^{-x_1\cdot \frac{2}{r}} \cdot 2 \cdot f((x_2,x_3,\ldots,x_{l+1}),r) \\
		    &\leq  x_1 \cdot \frac{2}{r} + e^{-x_1\cdot \frac{2}{r}} \cdot 10 \leq 5 \quad\rhd\text{since $1\leq x_1 \cdot \frac{2}{r}\leq 2$}
	\end{align*}
\end{proof}

We are now ready to prove \cref{Lemma:first_iteration}.
\begin{proof}[Proof of \cref{Lemma:first_iteration}]
	We have $(s/2) \cdot |\cS_{\cvr,1}| \leq \sum_{e \in \cE} X_e$ as during the first stage,  \cref{alg:base} only adds sets which have at least $s/2$ free elements at the beginning of the iteration. As $\ee{\sum_{e \in \cE}X_e} \leq 5\cdot n$, we get $\ee{|\cS_{\cvr,1}|} = O(n/s)$ as desired.
\end{proof}

\section{A Generic Algorithm}
\label{sec:generic-algorithm}
In this section we present~\cref{alg:generic-alg}. The algorithm differs from~\cref{alg:base} in two ways. 
First, it only \emph{estimates} the number of uncovered elements contained in a given set. 
Second, it fixes all the randomness right at the beginning of \cref{alg:generic-alg}. After \cref{line:gen-define-Sik}, the algorithm just executes a deterministic process. 
Both of these modifications are crucial to get efficient LCAs and will be explained in more detail in~\cref{section:generig-estimating}. 

Although \cref{alg:generic-alg} estimates set sizes instead of computing them exactly (as \cref{alg:base} does) it is, however, not clear how to simulate \cref{alg:generic-alg} in the local computation model with significantly less than $(st)^{O(\log s \cdot \log t)}$ many queries. The main reason for this is the following. The efficiency of a LCA is measured by its worst-case query complexity. In our case, this means that in order for a (randomized) LCA to have a query complexity of $T(s,t)$, each set $S$ needs to decide after at most $T(s,t)$ queries whether it is contained in the set cover or not with high probability (in $n$). As we aim for a query complexity independent of $n$, this essentially implies that for each set $S$, $S$ needs to decide after at most $T(s,t)$ queries whether it is part of the set cover or not, even if all the random bits are chosen adversarially.
Although the subsequent LCAs will essentially simulate \cref{alg:generic-alg}, crucially, they will bound the influence that ``bad randomness'' can have on the query complexity by, e.g., immediately adding sets to the set cover for which the number of sampled elements is too large. However, these tests make it hard to directly establish an approximation guarantee for those algorithms. To deal with that problem, we first establish the approximation guarantee for \cref{alg:generic-alg} in \cref{section:generic-approx}. After that, we use the approximation guarantee of \cref{alg:generic-alg} to establish the approximation guarantee of the subsequent algorithms by relating them to \cref{alg:generic-alg}.

\begin{algorithm}
		\caption{\label{alg:generic-alg} A generic set cover algorithm}
		
	Let $(\cS,\cE)$ be some set cover instance.
		
	\For{each pair of $(i,k)$ where $i \in [\log s]$ and $k \in [\log t]$ \label{line:generic-loop-Bi(S)}}{
	$B_i(S) \leftarrow$ a subset of $S$ s.t.~each element is included indep. w.p.~$p_i = \min(1,\logfive \cdot {2^i \over s})$. \label{line:gen-define-Bi-and-pi}

	$\cS_{i,k} \gets$ a subfamily of $\cS$ s.t.~each set is included indep. w.p.~$2^k/t$.
	\label{line:gen-define-Sik}
	}
		
	\label{line:gen-end-for-loop}			
	\For{stage i = 1 to $\log s$}{				
			\For{iteration $k = 1$ to $\log t $}{			
				$\degree(S) \eqdef \frac{|B_i(S) \cap \cE_{i,k}|}{p_i}$. \label{line:define-degree-S}\tcp{$\cE_{i,k}$ denotes the set of free elements in the beginning of current iteration}				
				\For{each set $S \in \cS_{i,k}$ with $\degree(S) \geq s/2^i$ \label{line:test-degree-S}}
				{{\bf add} $S$ to $\cS_{\cvr}$.}
			}			
	}
	\Return $\cS_{\cvr}$
\end{algorithm}

\subsection{Details of \cref{alg:generic-alg}}  \label{section:generig-estimating}
In this section, we further elaborate on the design decisions made in \cref{alg:generic-alg}. 
For a given set $S$, \cref{alg:generic-alg} estimates the number of free elements within stage $i$ by counting the number of uncovered elements in $B_i(S)$. 
The set $B_i(S)$ is a random subset of $S$ such that each element in $S$ is contained in $B_i(S)$ with a probability of $p_i$.  Note that if $p_i = 1$, then $\degree(S)$ is equal to the actual number of free elements of $S$ at the beginning of iteration $k$ in stage $i$. 
Otherwise, if $p_i < 1$, $\degree(S)$ estimates the number of free elements to be at least $s/2^i$ if the number of uncovered elements in $B_i(S)$ is at least $\logfive$ and otherwise $\degree(S)$ estimates that $S$ has fewer than $s/2^i$ free elements. 
Note that if we would not reuse the samples within the same stage, then $\degree(S)$ would be an unbiased estimator of $d(S)$. However, reusing the samples will be helpful in the analysis as we get the following monotonicity property within stage $i$: $S$ does not get added to the set cover in stage $i$ once the estimated number of uncovered elements in $S$ is smaller than $s/2^i$. This fact will turn out to be very useful in order to establish the approximation guarantee of \cref{alg:generic-alg}. 

Note that we only need to decide on $B_i(S)$ at the beginning of stage $i$ and on $\cS_{i,k}$ at the beginning of iteration $k$ in stage $i$. The reason why we have decided to fix all the $B_i(S)$'s and all the $\cS_{i,k}$'s at the beginning is to make the connection of \cref{alg:generic-alg} to the subsequent algorithms more obvious.

\subsection{Approximation Analysis of \cref{alg:generic-alg}} \label{section:generic-approx}
In this section we prove the following approximation guarantee for \cref{alg:generic-alg}.
\begin{theorem} \label{theorem:gen-set-cover-size}
	Let $\cS_{\cvr}$ be the solution returned by \cref{alg:generic-alg}. Then, $\ee{|\cS_{\cvr}|} = O(\log s) \cdot \opt$.
\end{theorem}
Now we give a roadmap of our analysis. In \cref{line:gen-define-Sik}, we include each set $S$ in $\cS_{i,k}$ with a probability of ${2^k}/{t}$. If $S$ is not contained in $\cS_{i,k}$, then $S$ will not be added to the set cover in iteration $k$ of stage $i$. If $S$ is contained in $\cS_{i,k}$, then $S$ will be added to the set cover in iteration $k$ of stage $i$ if the estimated number of free elements in $S$ at the beginning of iteration $k$ in stage $i$ is at least ${s}/{2^i}$. 
In this section we show that $\ee{\cS_{\cvr}} = O(\log s) \cdot \OPT$.  To that end, we upper-bound the expected size of
\begin{itemize}
	\item $\setsuperi$: Sets which contain at least $4 \cdot \frac{s}{2^i}$ free elements at the beginning of stage $i$. 
	\item  $\setssmall$: Sets which contained fewer than $\frac{1}{2} \cdot \frac{s}{2^i}$ free elements in the iteration they got added to the set cover
	\item  $\setsnormal$: Sets which contained between $\frac{1}{2} \cdot \frac{s}{2^i}$ and $4 \cdot \frac{s}{2^i}$ free elements in the iteration they got added to the set cover
\end{itemize}
Bounding $\ee{\setsuperi}$ and $\ee{\setssmall}$ is straightforward. Let $n_i$ denote the number of free elements at the beginning of stage $i$ and let $s_i := s/2^i$. We show that $\ee{\setsnormal} = O(n_i/s_i)$ in a similar way as in the approximation analysis of \cref{alg:base} by relying on the monotonicity property of the estimates established in the previous subsection. However, unlike~\cref{alg:base}, the bound $\OPT = \Omega(n_i/s_i)$ does not generally hold in \cref{alg:generic-alg}, as we are loosing the  guarantee that all sets contain no more than $2 \cdot s_i$ uncovered elements at the beginning of stage $i$. Nevertheless, we can still relate $\ee{\setsnormal}$ to $\OPT$ by using the bound on $\ee{\setsuperi}$.

\begin{lemma}  \label{lemma:gen-small-too-big}
	Let $S$ be some arbitrary set and $i$ be some arbitrary stage. The probability that $S$ gets added to the set cover in some arbitrary iteration of stage $i$ such that $S$ had fewer than $\frac{1}{2} \cdot s/2^i$ many free elements at the beginning of the iteration in which it was added is at most $\negfive$.
\end{lemma}
\begin{lemma} \label{lemma:gen-big-too-small}
	Let $S$ be some arbitrary set and $i$ be some arbitrary stage. The probability that $S$ has more than $2 \cdot s/2^i$ many free elements at the end of stage $i$ is at most $\negfive$.
\end{lemma}
\begin{proof}[Proof of~\cref{lemma:gen-small-too-big} and~\cref{lemma:gen-big-too-small}]
	Note that we can assume that $p_i < 1$ as otherwise the estimated number of free elements of a set as computed in~\cref{line:define-degree-S} is equal to the actual number of free elements of the set. In this case, both \cref{lemma:gen-small-too-big} and \cref{lemma:gen-big-too-small} trivially hold as the bad events cannot happen.
	 
	Next, note that the set of random bits used in~\cref{alg:generic-alg} are set right at the beginning. 
After~\cref{line:gen-end-for-loop} the algorithm basically executes a deterministic process. By the principle of deferred decision making, assume that we fix all the random bits except for those ones determining $B_i(S)$. Let $k$ be the first iteration for which $S$ is contained in $\cS_{i,k}$. Since $p_{i,\log t} = 1$, there always exists such an iteration. 

Since the algorithm uses the same set of elements for estimation within one stage, if in iteration $k$, the number of free elements in $S$ is estimated to be less than $s/2^i$, then $S$ will not be added to the set cover within stage $i$ in any subsequent iteration. Furthermore, if the number of free elements in $S$ is estimated to be at least $s/2^i$ in iteration $k$, then $S$ will be added to the set cover and after that the estimated number of free elements will always be 0.

	This implies that in order to prove~\cref{lemma:gen-big-too-small}, we can assume that $S$ has less than $\frac{1}{2} \cdot {s\over 2^i}$ many free elements at the beginning of iteration $k$ in stage $i$ and we only need to show that the probability that the estimated number of free elements in $S$ is at least $s/2^i$ is $\negfive$. Note that the first condition implies that the expected number of free elements in $B_i(S)$ is at most $\frac{1}{2} \cdot \logfive$. We only estimate $S$ to have a size of at least $s/2^i$ if $B_i(S)$ contains at least $\logfive$ many free elements. A Chernoff bound (\cref{lemma:chernoff}-\ref{item:delta-at-least-1}) implies that this happens with a probability of $\negfive$. 
	
	Similarly, to prove~\cref{lemma:gen-small-too-big}, we can assume that $S$ has at least $2 \cdot s/2^i$ many free elements at the beginning of iteration $k$ in stage $i$ and we need to bound the probability that the estimated number of free elements of $S$ is less than $s/2^i$. A Chernoff bound (\cref{lemma:chernoff}-\ref{item:delta-at-most-1-le}) implies that this happens with a probability of $\negfive$.
\end{proof}

\begin{lemma}  \label{lemma:gen-bound-small-sets}
	Let $\cS_{\sml}$ be the family of sets containing those sets which were added to the set cover in some arbitrary stage $i$ despite having fewer than $\frac{1}{2} \cdot {s\over 2^i}$ many free elements at the beginning of the iteration in which they were added. Then, $\ee{|\cS_{\sml}|} = O(\opt)$.
\end{lemma}
\begin{proof}
	By \cref{lemma:gen-small-too-big} and a union bound over the $\log s$ many stages, any set $S$ is in $\cS_{\sml}$ with a probability of $\log s \cdot \negfive$. The lemma follows as there are at most $m\leq n \cdot t$ many sets and $\opt\geq n/s$.
\end{proof}
\begin{lemma} \label{lemma:gen-constant-coverage}
	Consider some arbitrary stage $i$. Let $e$ be an arbitrary element and $\xei$ the random variable which equals $0$ if $e$ does not get covered for the first time during stage $i$, and otherwise is equal to the number of sets that contain $e$ and were added to the set cover in the same iteration in which $e$ was covered for the first time. Then, $\ee{\xei} \leq 5$. Furthermore, this bound even holds if we fix all the randomness needed for all stages up to stage $i$ in some arbitrary way.
\end{lemma}
\begin{proof}
	Let $\cS_e$ be the family of sets that contain $e$. We show this lemma by using the principle of deferred decision making. Suppose that we have fixed all the randomness used in stage $i$ other than the sets in $\cS_e$ at the beginning: e.g. it is not yet determined whether $S \in \cS_{i,k}$ or not for $S \in \cS_e$. We will fix those decisions only in the iteration for which we need them. Let $N_{i,k}$ denote the random variable (with respect to the randomness which we have not fixed) that is equal to the number of sets in $\cS_e$ whose estimated size is at least $s/2^i$ at the beginning of iteration $k$ in stage $i$. As we always use the same elements to estimate the size of a particular set within stage $i$, we have $t \geq N_{i,1} \geq N_{i,2} \geq ... \geq N_{i,\log t}$. As we have fixed all the randomness needed to execute all the stages up to stage $i$, $N_{i,1}$ is just equal to some number $n_{i,1}$. 
	Now, each of those $n_{i,1}$ sets is considered in the first iteration in stage $i$ with probability $\frac{2}{t}$. Thus, the number of sets which cover $e$ in the first iteration is a random variable $Y_1 \sim \mathrm{Bin}(n_{i,1},p_{i,1})$. 
	If $Y_1 > 0$, then $\xei = Y_1$. Otherwise, none of the sets in $\cS_e$ is added to the set cover. Given that information, $N_{i,2}$ is equal to some number $n_{i,2}$. Then, the number of sets containing $e$ in the second iteration is a random variable $Y_2 \sim \mathrm{Bin}(n_2,p_{i,2})$ and so on. Note that this process is completely analogous to the process in the proof of \cref{Lemma:constant coverage}. Thus, we can conclude that $\ee{\xei} \leq 5$. 
\end{proof}
\begin{lemma} \label{lemma:gen-super}
	Let $\setsuperi$ be a family of sets that contains those sets which have more than $4 \cdot {s\over 2^i}$ many free elements at the beginning of stage $i$. Then, $\ee{|\setsuperi|} \leq \frac{n}{t \cdot s^3}$ and $\Pr[|\setsuperi| > \frac{n}{s^2}] \leq \frac{1}{st}$. 
\end{lemma}
\begin{proof}
	Consider an arbitrary set $S$. For $i = 1$, $S$ is never contained in $\setsuperi$. For $i > 1$, $S$ is only in $\setsuperi$ if $S$ has at least $2 \cdot (s/2^{i-1})$ many free elements at the end of stage $i-1$. By~\cref{lemma:gen-big-too-small}, this only happens with a probability of $\negfive$. As we have at most $t \cdot n$ many sets, $\ee{|\setsuperi|} \leq \frac{n}{t \cdot s^3}$. Hence, a Markov bound implies that $\Pr[|\setsuperi| > \frac{n}{s^2}] \leq \frac{1}{st}$. 
\end{proof}

\begin{lemma}  \label{lemma:gen-bound-big-sets}
	Consider the beginning of stage $i$. Let $\setsnormal$ be the family of sets consisting of those sets which are added to the set cover in some arbitrary stage $i$ and additionally, at the beginning of the iteration in which they are added to the set cover, contain at least $\frac{1}{2} \cdot (s/2^i)$ and at most $4 \cdot (s/2^i)$ many free elements. 
	Suppose that the randomness needed for all stages up to $i$ are fixed in advance. Let $n_i$ denote the number of free elements at the beginning of stage $i$ and $s_i := s/2^i$. Then, $\ee{|\setsnormal|} = O(n_i/s_i)$.
\end{lemma}
\begin{proof}
	Let $\cE_i$ denote the set of uncovered elements at the beginning of stage $i$. As all sets in $\setsnormal$ have at least $\frac{1}{2} \cdot (s/2^i)$ many free elements at the beginning of the iteration in which they are added to the set cover,
	$$|\setsnormal| \leq \frac{1}{\frac{1}{2} \cdot s_i} \cdot \sum_{e \in \cE_i} \xei$$
	Thus, as $|\cE_i| = n_i$ and $\ee{\xei} \leq 5$,
	$$ \ee{|\setsnormal|} \leq \frac{5 \cdot n_i}{\frac{1}{2} \cdot s_i} = O(n_i/{s_i})$$
\end{proof}
\begin{lemma}\label{lemma:generic-alg-bound-on-setsnormal}
		Consider the beginning of stage $i$. Let $\setsnormal$ be the family of sets that are added to the set cover in stage $i$ and additionally, at the beginning of the iteration in which they are added to the set cover, contain at least $\frac{1}{2} \cdot (s/2^i)$ and at most  $4 \cdot (s/2^i)$ many free elements. Then, $\ee{|\setsnormal|} = O(\opt)$.
\end{lemma}
\begin{proof}

	Assume that we are at the beginning of stage $i$. As before, we denote with $n_i$ the number of free elements at the beginning of stage $i$ and $s_i := s/2^i$. As a reminder, $\setsuperi$ was defined as the sets which have more than $4 \cdot s/2^i$ many free elements at the beginning of stage $i$. We have the following lower bound for $\opt$
	\[
		\opt \geq \frac{n_i - s \cdot |\setsuperi|}{4 \cdot s_i}.
	\]
	
	This holds because after adding all sets in $\setsuperi$ to the set cover, at least $n_i - s \cdot |\setsuperi|$ free elements remain and the maximal number of free elements any set contains is at most $4 \cdot s_i$. 
	Assume that $|\setsuperi| \leq \frac{n}{s^2}$. Then
	\[
		\opt \geq  \frac{n_i - s \cdot |\setsuperi|}{4 \cdot s_i} \geq \Omega(n_i/s_i) - \frac{n}{s} \geq \Omega(n_i/s_i) - \opt.
	\]
	Hence, in this case $\opt = \Omega(n_i/s_i) \geq \Omega(\ee{|\setsnormal|})$. (Note that $|\setsuperi|$ only depends on the randomness needed for all stages prior to stage $i$ and therefore we can apply the previous lemma) \\
	By \cref{lemma:gen-super}, $\Pr[|\setsuperi| > \frac{n}{s^2}] \leq \frac{1}{s \cdot t}$.  Therefore
	\begin{align*}
		\ee{\setsnormal} 
		&=\ee{|\setsnormal|\ |\ |\setsuperi| > \frac{n}{s^2}} \cdot \prob{|\setsuperi| > \frac{n}{s^2}} \\
		&+ \ee{|\setsnormal|\ |\ |\setsuperi| \leq \frac{n}{s^2}} \cdot \prob{|\setsuperi| \leq \frac{n}{s^2}} \\
		&\leq n \cdot t \cdot \frac{1}{s \cdot t} + \ee{|\setsnormal|\ |\ |\setsuperi| \leq \frac{n}{s^2}} \\
		&\leq O(n/s) + O(\opt) \\
		&= O(\opt).
	\end{align*}
\end{proof}

We are now ready to prove the main result of this section.
\begin{proof}[Proof of \cref{theorem:gen-set-cover-size}]
	First, observe that
	\[
		\ee{|\cS_{\cvr}|} \leq \ee{|\cS_{\sml}|} + \sum_{i=1}^{\log s} \ee{|\setsuperi|} + \ee{|\setsnormal|}.
	\]
	Then, from \cref{lemma:gen-bound-small-sets,lemma:gen-super,lemma:generic-alg-bound-on-setsnormal} we conclude
	\begin{align*}
		\ee{|\cS_{\cvr}|} &\leq O(n/s) + \sum_{i=1}^{\log s} \frac{n}{t \cdot s^3} + O(\opt) \\
			&\leq O(\log s) \cdot \opt.
	\end{align*}
\end{proof}

\subsubsection{Bounding the Probability of ``Bad'' Events}

In this section we introduce the notion of {\em bad elements} and {\em bad sets}. Bounding the probability of some element or set being {\em bad} will play a central role in relating all the subsequent algorithms to \cref{alg:generic-alg}. The connection we get is the following: Consider that we are executing \cref{alg:generic-alg} and some subsequent algorithm with the same randomness. Let $S$ be some arbitrary set and assume that there does not exist a bad element or a bad set in the $(c \cdot \log s \cdot \log t)$-hop neighborhood of $S$ for some large enough constant $c$. Then, $S$ is either added to the set cover in both algorithms or in none of them. As the probability is fairly large that there does not exist such a bad element or bad set in the neighborhood of $S$, by the approximation analysis of~\cref{alg:generic-alg}, we conclude that the subsequent algorithm returns a set cover with an expected size of $O(\log s) \cdot \OPT$.

	An element $e$ is a {\em bad element} iff there exist $i \in [\log s]$ and $1\leq k_1\leq k_2\leq \log t$ such that the following conditions hold:
	\begin{enumerate}
		\item $e$ is uncovered at the beginning of iteration $k_1$ of stage $i$.
		\item $e$ is contained in more than $\logten \cdot 2^{k_2-k_1}$ sets in $\cS_{i,k_2}$ with an estimated size of at least $s/2^i$ at the beginning of iteration $k_1$ of stage $i$.
	\end{enumerate}

\begin{lemma} \label{lemma:gen-bad-element}
	For any element $e$, $\Pr[e \text{ is bad }] = \negfour$.
\end{lemma}
\begin{proof}
First consider the case $k_1 = 1$. The expected number of sets in $\cS_{i,k_2}$ which contain $e$ is at most $t \cdot \frac{2^{k_2}}{t} = 2^{k_2}$. Hence, an application of Chernoff bound (\cref{lemma:chernoff}-\ref{item:delta-at-least-1}) implies that the probability, that the number of sets in $\cS_{i,k_2}$ exceeds $\logten \cdot 2^{k_2-k_1}$, is at most $\negfive$. \\
Next, consider the case $k_1 > 1$ and we consider two subcases for this case. 
\begin{itemize}
	\item{\bf There are at least $\logfive \cdot \frac{t}{2^{k_1-1}}$ many sets that contain $e$ and have estimated size of at least $s/2^i$ at the beginning of iteration $k_1 - 1$ in stage $i$.} Then, the probability that $e$ remains uncovered at the beginning of iteration $k_1$ is at most 
\begin{align*}
\rb{1-\frac{2^{k_1-1}}{t}}^{\logfive \cdot \frac{t}{2^{k_1-1}}} = \negfive. 
\end{align*}
\item{\bf There are less than $\logfive \cdot \frac{t}{2^{k_1-1}}$ many sets that contain $e$ and have estimated size of at least $s/2^i$ at the beginning of iteration $k_1 - 1$ in stage $i$.} 
This implies that the expected number of sets in $S_{i,k_2}$ that contain $e$ and have estimated size of at least $s/2^i$ at the beginning of iteration $k_1$ of stage $i$ is at most $\logfive \cdot \frac{t}{2^{k_1-1}} \cdot \frac{2^{k_2}}{t}$. Hence, an application of Chernoff bound (\cref{lemma:chernoff}-\ref{item:delta-at-least-1}) implies that the number of sets in $\cS_{i,k_2}$ that contain $e$ and have an estimated size of at least $s/2^i$ at the beginning of iteration $k_1$ in stage $i$ exceeds $\logten \cdot 2^{k_2-k_1}$ with a probability of at most $\negfive$. The lemma follows by a union bound over the at most $O(\log s \cdot \log^2 t)$ many choices for $i$, $k_1$ and $k_2$.
\end{itemize}
\end{proof}
A set $S$ is a {\em bad set} iff there exist $1\leq i_1\leq i_2\leq \log s$ such that the number of free elements in $B_{i_2}(S)$ at the beginning of stage $i_1$ exceeds $\logten \cdot 2^{i_2-i_1}$.
\begin{lemma} \label{lemma:gen-bad-set}
 For any set $S$, $\Pr[S \text{ is bad }] = \negfour$.
\end{lemma}
\begin{proof}
We first consider the case that $S$ has more than $2 \cdot s/2^{i_1-1}$ many free elements at the beginning of stage $i_1$. \cref{lemma:gen-big-too-small} implies that this happens with a probability of $\negfive$.
Next, suppose that $S$ has less than $2 \cdot s/2^{i-1}$ free elements at the beginning of stage $i_1$. This implies that the expected number of free elements in $B_{i_2}(S)$ is at most $2 \cdot s/2^{i_1-1} \cdot p_{i_2} \leq 4 \cdot \logfive \cdot 2^{i_2-i_1}$. Hence, an application of Chernoff bound (\cref{lemma:chernoff}-\ref{item:delta-at-least-1}) implies that the number of free elements in $B_{i_2}(S)$ is more than $\logten \cdot 2^{i_2-i_1}$ with a probability of at most $\negfive$. The lemma follows by a union bound over the $O(\log^2(s))$ many possible choices for $i_1$ and $i_2$.
\end{proof}

	Consider the bipartite graph $G$ induced by the set system $(\cS, \cE)$. A set $S$ (resp., An element $e$) has a {\em bad neighborhood} iff there exists a bad element (resp., a bad set) in the $10$-hop neighborhood of $S$ (resp., $e$) in $G$. 
\begin{lemma} \label{lemma:gen-bad-neighborhood}
The probability that an arbitrary set or element has a bad neighborhood is at most $\negthree$.
\end{lemma}
\begin{proof}
	The lemma directly follows by \cref{lemma:gen-bad-element} and \cref{lemma:gen-bad-set} together with a union bound over the at most $\poly(st)$ many sets and elements in the $10$-hop neighborhood of $S$.
\end{proof}

\newcommand{\numel}{\log^2 t \cdot \log^2 s}
\newcommand{\numsets}{4 \cdot \log^2 s \cdot 2^{3T}}
\newcommand{\hsil}{\widehat{\cS}_{i,k}}
\newcommand{\hsjk}{\widehat{\cS}_{i,j \cdot T + \ell}}

\section{Beyond the Parnas-Ron Paradigm}
\label{sec:sqrt-log-t}

By applying the standard reduction of Parnas and Ron~\cite{parnas2007approximating} for transforming distributed algorithms into LCAs, \cref{alg:base} and \cref{alg:generic-alg} can be simulated in the local computation model with a query complexity of $(st)^{O(\log{s} \log{t})}$. In this section, we present an LCA that can be implemented with a query complexity of $(st)^{O\rb{\log{s} \sqrt{\log{t}}}}$. 

\subsection*{Overview of Our Approach}
We now discuss our general ideas for obtaining a query-efficient LCA simulation. These ideas will be applied to the algorithms developed in this section and in \cref{sec:log-log-t}. The main purpose of the discussion that follows is to illustrate a way to design LCAs for the set cover so that it suffices to access only a ``small'' number of sets and elements in the direct neighborhood of another set/element.

\paragraph{LCA from the point of view of a set.} Our LCA simulations test whether a given set $S$ is in the cover or not by recursively testing whether $S$ has been added to the cover in any of the rounds, starting from round $\log{s} \cdot \log{t}$. To be more detailed, consider a round $r$ of stage $i$. To decide whether $S$ has been added to the set cover by the end of round $r$ in \cref{alg:generic-alg}, we perform two steps: 
\begin{enumerate}[(1)]
	\item\label{item:test-S-added-to-the-cover} We test whether $S$ is added to the cover during the first $r - 1$ rounds of the algorithm.
	\item\label{item:test-S-not-added-size-of-BiS} If $S$ has not been added to the set cover yet, we count the number of yet uncovered elements in $B_i(S)$ at the beginning of round $r$ and proceed based on the count.
\end{enumerate}
To perform step~\eqref{item:test-S-added-to-the-cover}, it suffices to perform the same two steps for only round $r-1$. For step~\eqref{item:test-S-not-added-size-of-BiS}, as we already know the outcome of step~\eqref{item:test-S-added-to-the-cover}, we expect that $S$ does not have more than $10 \cdot s/2^i$ many free elements at the beginning of stage $i$. As each element of $S$ is contained in $B_i(S)$ with a probability of  $\min\rb{1,\logfive \cdot \tfrac{2^i}{s}}$, we furthermore expect that $B_i(S)$ does not have more than $\poly(\log s\log t)$ many free elements at the beginning of stage $i$, and therefore at the beginning of round $r$. Intuitively, the sooner we detect those uncovered elements, the more efficient LCA we are likely to obtain. The reason is that testing whether an element is free or not requires performing additional recursive tests. Since in general $B_i(S)$ may contain a lot of elements that are covered before stage $i$, these recursive tests for each covered element can be very expensive query-wise. To alleviate that, informally speaking, we remove all previously covered elements  at the beginning of a stage $i$ (see \cref{alg:main-alg}) or at the beginning of some phase $j$ that consists of multiple iterations (see \cref{alg:sqrt-t}). Therefore, within a phase $j$ (or stage $i$), $S$ only needs to check for all elements in this ``sparsified'' version of $B_i(S)$, which only includes the free elements in $B_i(S)$, whether they are covered or not at the beginning of round $r$. We expect this ``sparsified'' version of $B_i(S)$ to only consist of $O(\poly(\log s \log t))$ many elements.

\paragraph{LCA from the point of view of an element.} In the following, assume that round $r$ corresponds to iteration $k$ in stage $i$. Furthermore, assume that we merge iterations into multiple phases, with iteration $k$ being part of phase $j$. In order to test whether an element $e$ is covered for the first time in round $r$ of \cref{alg:generic-alg}, our LCA simulation performs the following test:
\begin{itemize}
	\item Given that $e$ was not yet covered, does {\em any} of the sets $e$ is contained in get added to the set cover in round $r$?
\end{itemize}
To answer this question efficiently, ideally, $e$ should invoke tests on a small number of sets. We develop the following two ideas to reduce the number of sets that the above question should be answered for:
\begin{enumerate}[(i)]
	\item As in \cref{alg:generic-alg}, we sample all sets that are potentially considered in stage $i$ of iteration $k$ beforehand. This sample is denoted by $\cS_{i,k}$.
	\item At the beginning of phase $j$, we remove all sets from $\cS_{i,k}$ that have an estimated set size smaller than $s/2^i$. 
\end{enumerate} 
These ideas have the following positive consequence. Assume that each phase has length $T$ and that round $r$ corresponds to the $\ell$-th iteration within phase $j$. If $e$ is still uncovered at the beginning of phase $j$, then we expect $e$ to be contained in fewer than $10 \cdot \frac{t}{2^{j \cdot T}}$ many ``large'' sets, as each ``large'' set is added to the set cover with a probability of $\frac{2^{j \cdot T - 1}}{t}$ in  iteration $j \cdot T - 1$. Each of those large sets is contained in $\cS_{i,k}$ with a probability of $\frac{2^{j\cdot T+l}}{t}$. Hence, we expect that the number of sets containing $e$ and belonging to the sparsified version of $\cS_{i,k}$ is at most $2^{l} \leq 2^T$, which for our choice of $T$ is significantly less than $t$.

\paragraph{Handling expectations.} In the previous paragraphs, most of the statements we provided on the sparsified neighborhoods of elements and sets were stated only in expectation. However, as noted in \cref{sec:generic-algorithm}, the query complexity still needs to be guaranteed even if the random bits are chosen adversarially. To guarantee this, we immediately add a set to the set cover if the sparsified version of $B_i(S)$ contains much more elements than expected at the beginning of a phase (\cref{alg:sqrt-t}) or stage (\cref{alg:rec-split}). Likewise, we ``pretend'' that an element is covered if $e$ is contained in much more sets than expected in the sparsified version of $\cS_{i,k}$. This is also the only difference to \cref{alg:generic-alg}. As $e$ reports to be covered, it might happen that $e$ is still uncovered at the end of the algorithm. In that case, we simply add one set to the set cover that contains $e$. In this way, bad randomness can only affect the approximation guarantee, but not the query complexity of the LCA. In this way, intuitively, while not completely true, one can think about executing one phase in a sparsified set system where the maximal set size is $O(\poly(\log s \log t))$ and the maximal number of sets any element is contained in is $2^{O(T)}$.  Note that there is a trade-off for the phase length $T$. With increasing $T$, the sparsified versions are becoming less and less sparse. Likewise, in order to produce the illusion of a sparsified set system, we need to query the set system corresponding to the previous phase $\poly(st)$ many times for each query that we receive for the sparsified set system. We balance this tradeoff by  setting $T = \sqrt{\log t}$.

\subsection{Details about \cref{alg:sqrt-t}}
In this section we design \cref{alg:sqrt-t}. Compared to \cref{alg:generic-alg}, \cref{alg:sqrt-t} processes iterations in groups of $T$. Each group is called a \emph{phase}. \cref{alg:sqrt-t} also samples sets $B_i(S)$, for each stage $i$ and each set $S$. At the beginning of phase $j$, we compute the set $B_{i,j}(S)$ by removing all elements in $B_i(S)$ that are covered at the beginning of phase $j$. The LCA can efficiently estimate the number of free elements within phase $j$ by only counting the number of free elements in $B_{i,j}(S)$. Furthermore, $\hcS_{i,k}$ denotes the family of sets obtained by removing all the sets from $\cS_{i,k}$ that have an estimated number of free elements smaller than $s/2^i$ at the beginning of the respective phase.

\begin{algorithm}[t]
	$T \eqdef \sqrt{\log(t)}$
	
	\For{each pair of $(i,k)$ where $i \in [\log s]$ and $k \in [\log t]$}{
		
		$B_i(S)$ and $\cS_{i,k}$ are generated in the exact same way as in \cref{alg:generic-alg} 
		
	}

	\For{ $i = 1$ to $\log s$
	 \label{line:sqrt-t-stage-loop}}{
		\For{phase $j = 0$ to  $T-1$}{
			\For{each set $S \in \cS$}{
				$B_{i,j}(S) \gets B_i(S) \cap \cE_{i,j}$ 
				\tcp{$\cE_{i,j}$ denotes the set of free elements at the beginning of the phase}
				\If(\tcp*[h]{test whether $S$ is a bad set}){$|B_{i,j}(S)| \geq \logten$}{
					{\bf add} $S$ to $\cS_{\cvr}$.  \label{line:alg2_large_set}
				}
			}

			\For{each $k \in \{j \cdot T + 1,j \cdot T+2, ...,(j+1) \cdot T\}$}{		    	
				$\hsil \gets \{S \;|\; S \in \cS_{i,k} \text{ and } \degree(S) \geq s/2^i\}$ \\
				\For{each $e\in \cE_{i,j}$ in parallel\label{line:loop-sqrt-t-discard-high-degree-els}}{
					\If(\tcp*[h]{test whether $e$ is a bad element})
					{$d_{\hsil}(e) \geq \logten \cdot 2^T$\label{line:sqrt-t-if-high-degree-element}}{				   			
						{\it pretend} that $e$ is covered. 
				   		\tcp{$e$ will be covered in~\cref{line:sqrt-t-clean-up}.}\label{line:sqrt-t-pretend}
				   	}
				}
			}

			\label{line:sqrt-t-after-filtering}

			\For{iteration $\ell$ = 1 to $T$ \label{line:sqrt-t-the-most-inner-loop}}{
				\For{each set $S \in \hsjk$ in parallel} {
		    		\If{$\degree(S) \geq s/2^i$\label{line:estimate-set-test}}{ 
		  			  	{\bf add} $S$ to $\cS_{\cvr}$. \label{line:alg2_add_S}
	    			}

				}

			} 
		}
	}

	\For(\tcp*[h]{Handling bad elements}){each free element in parallel $e$}{
		{\bf add} the set with smallest ID which contains $e$ to $\cS_{\cvr}$.\label{line:sqrt-t-clean-up}
	}	
	\caption{A variant of \cref{alg:generic-alg} that can be simulated by an LCA with $(st)^{O(\log{s} \sqrt{\log{t}})}$ many queries. \label{alg:sqrt-t}}
\end{algorithm}

\newcommand{\oracleBij}{\cO_{B_{i,j}}(S)}
\newcommand{\oracleSij}{\cO_{\hat{S}_{i,j}}(e)}

\subsection{Simulation of \cref{alg:sqrt-t} in LCA}

\label{sec:query-complexity-of-alg-sqrt-t}
In the rest of this section we describe the LCA simulation of \cref{alg:sqrt-t} and analyze its query complexity. 
Our final goal is to provide oracle access to the set cover produced by \cref{alg:sqrt-t}. 
We call this oracle $\Oracle$. That is, $\Oracle(S)$ answers if $S$ is part of the set cover or not. 
It is convenient to define intermediate oracles for the analysis. 
Namely, $\Oracle_i(S)$ outputs if $S$ is part of the set cover at the end of the $i$-th stage, $\Oracle_{i,j}(S)$ outputs if $S$ is part of the set cover at the end of phase $(i,j)$ and $\Oracle_{i,j,\ell}(S)$ outputs if $S$ is part of the set cover at the end of iteration $(i,j,\ell)$. We remark that $\Oracle_{i,j,0}(S)$ is defined to output whether $S$ is added to the set cover at the beginning of iteration $(i,j,1)$. 
Note that the oracle $\Oracle_{i,j,T}$ is always equal to the oracle $\Oracle_{i,j}$ and $\Oracle_{i,T-1}$ is always equal to $\Oracle_i$. However, $\Oracle_{\log(s)}$ is not equal to $\Oracle$ because of \cref{line:sqrt-t-clean-up}. We also define oracles $\OracleEl_i(e)$, $\OracleEl_{i,j}(e)$ and $\OracleEl_{i,j,\ell}(e)$ which answer if $e$ was covered at the end of stage $i$, phase $(i,j)$ or iteration $(i,j,\ell)$, respectively. Note that they also answer that $e$ is covered even if $e$ only pretends that it is covered. We also remark that $\OracleEl_{i,j,0}(e)$ is defined to output whether $e$ is covered by any of the sets at the beginning of iteration $(i,j,1)$. 

\paragraph{Description of $\Oracle_{i, j}$ and $\OracleEl_{i, j}$ given query access to $(\cS_{i,j}, \cE_{i,j})$.}
In this section, we assume to have query access to $(\cS_{i,j}, \cE_{i,j})$ which respectively denote the collection of unselected sets and uncovered elements at beginning of phase $(i,j)$. Given that, we describe how to answer oracles $\Oracle_{i, j}(S)$ and $\OracleEl_{i, j}(e)$ for some set $S$ or element $e$ by querying $(\cS_{i,j}, \cE_{i,j})$ at most $\poly(st)$ times.
	We first note that it takes $O(1)$ many queries to decide for a set $S$ and some $k \in [\log t]$ if $S \in \cS_{i,k}$: $S$ is contained in $\cS_{i,k}$ with a probability of ${2^{k} \over t}$.
	
	Furthermore, for a given set $S$, we can get all elements in $B_{i,j}(S)$ with $O(s)$ many queries to $(\cS_{i,j},\cE_{i,j})$ and for a given element $e \in \cE_{i,j}$ and some $k \in \{j \cdot T + 1, ..., (j+1) \cdot T\}$, we can get all the sets in $\hcS_{i,k}$ that contain $e$ with $O(st)$ many queries to $(\cS_{i,j},\cE_{i,j})$. \\
	
	To answer $\Oracle_{i, j, \ell}(S)$, we first check whether $S$ was already added to the set cover at the beginning of iteration $(i,j,\ell)$ or not:
\begin{enumerate}
\item{$\boldsymbol{\ell=0}$:} we simply check whether $S \notin \cS_{i,j}$
\item{$\boldsymbol{\ell>0}$:} we recursively invoke $\Oracle_{i, j, \ell-1}(S)$.
\end{enumerate}
	
If $S$ was indeed added, then we can immediately answer $\Oracle_{i, j, \ell}(S)$. Otherwise, we need to figure out whether $S$ gets added in iteration $(i,j,\ell)$ or not. To that end, we again consider two cases:

\begin{enumerate}
\item{$\boldsymbol{\ell=0}$:}  $S$  gets added to the set cover iff $|B_{i,j}(S)| \geq \logten$.
\item{$\boldsymbol{\ell>0}$:} we first check whether $S \in \cS_{i,j\cdot T+\ell}$. If $S \in \cS_{i,j\cdot T+\ell}$, then for each element $e \in B_{i,j}(S)$, we invoke the oracle $\OracleEl_{i, j, \ell-1}(e)$. If the estimated number of free elements in $S$ is at least $s/2^i$, then $S$ gets added to the set cover in iteration $(i,j,\ell)$, otherwise not. This step requires at most $\logten$ many recursive calls to the oracle $\OracleEl_{i,j, \ell-1}$.
\end{enumerate}
Next, we describe how to answer $\OracleEl_{i, j, \ell}(e)$. To do so, we first check whether $e$ was already covered (or pretends to be covered) at the beginning of iteration $(i,j,\ell)$. 
\begin{enumerate}
\item{$\boldsymbol{\ell=0}$:} we simply check whether $e\notin \cE_{i,j}$,
\item{$\boldsymbol{\ell>0}$:} we recursively invoke $\OracleEl_{i, j, \ell-1}(e)$.
\end{enumerate}
If $e$ is not covered, then we proceed as follows:
\begin{enumerate}
\item{$\boldsymbol{\ell=0}$:} we first invoke $\Oracle_{i,j,0}$ for all sets which contain $e$. If none of them is added to the set cover by iteration $(i,j,0)$, then we check if $e$ pretends to be covered in \cref{line:sqrt-t-pretend} by checking for each $k \in \{j \cdot T + 1, ..., (j+1)\cdot T\}$ if $e$ is contained in more than $\logten \cdot 2^T$ many sets in $\hcS_{i,k}$. 
\item{$\boldsymbol{\ell>0}$:} Otherwise, for all sets in $\hcS_{i,j \cdot T + \ell}$ containing $e$, we invoke the oracle $\Oracle_{i, j, \ell}$ to check if any of them is added to the set cover by the end of iteration $(i,j,\ell)$. If that's the case, then $e$ is covered after iteration $(i,j,\ell)$. This step requires at most $\logten \cdot 2^T$ recursive calls to $\Oracle_{i,j,\ell-1}$ oracle.
\end{enumerate}

	Let $Q_{set}(\ell)$ and $Q_{el}(\ell)$ be an upper bound on the total number of queries to $(\cS_{i,j},\cE_{i,j})$ that we need in order to answer  $\Oracle_{i, j, \ell}(S)$ and $\OracleEl_{i, j, \ell}(e)$ respectively. From the description, we derive the following recursions:
	\begin{align*}
		Q_{set}(0) &= O(\poly(st)) \\
		Q_{set}(\ell) & = Q_{set}(\ell-1) + \logten \cdot Q_{el}(\ell-1) + O(s) \\
					& \leq \poly(\log s) \cdot 2^{O(T)} \cdot (Q_{set}(\ell-1) + Q_{el}(\ell-1)) + O(st) \text{ for $\ell > 0$} \\
		Q_{el}(0) & = O(\poly(st)) \\
		Q_{el}(\ell) & = Q_{el}(\ell-1) +\logten \cdot 2^T \cdot Q_{set}(\ell) + O(st) \\
		& \leq \poly(\log s) \cdot 2^{O(T)} \cdot (Q_{set}(\ell-1) + Q_{el}(\ell-1)) + O(st) \text{ for $\ell > 0$}
	\end{align*}
	Therefore, we get $Q_{set}(\ell) = Q_{el}(\ell) = O(\poly(st))\cdot (\poly(\log s) \cdot 2^{O(T)})^\ell$. This implies $Q_{set}(T) = Q_{el}(T)  = O(\poly(st))$. Recall that we have $\Oracle_{i,j,T} = \Oracle_{i,j}$ and $\OracleEl_{i,j,T} = \OracleEl_{i,j}$. Thus, we can answer both $\Oracle_{1,0}$ and  $\OracleEl_{1,0}$ with $O(\poly(st))$ many queries.
	
\paragraph{Description of $\Oracle_{i,j}$ and $\OracleEl_{i,j}$ given query access to  $\Oracle_{i,j-1}$ and $\OracleEl_{i,j-1}$ ($j > 1$).} 
	It is easy to verify that we can give query access to $(\cS_{i,j},\cE_{i, j})$ by using $O(st)$ many queries to  $\Oracle_{i,j-1}$, $\OracleEl_{i,j-1}$ and the original set cover instance. 
	Together with the previous paragraph, this implies that we can answer $\Oracle_{i,j}(S)$ and $\OracleEl_{i,j}(e)$ with $O(\poly(st))$ many oracle calls to $\Oracle_{i,j-1}$, $\OracleEl_{i,j-1}$ and the original set cover instance.

\paragraph{Description of $\Oracle_{i}$ and $\OracleEl_{i}$ given query access to $(\cS_{i,0},\cE_{i,0})$.}
	The query complexity analysis of the previous case together with a simple induction argument implies that we can answer $\Oracle_{i,j}(S)$ and $\OracleEl_{i,j}(e)$ with $(st)^{O(j)}$ many queries to $(\cS_{i,0},\cE_{i,0})$. In particular, this implies that we can answer the oracles corresponding to the last phase in stage $i$, namely $\Oracle_{i,T-1}(S)$ and $\OracleEl_{i,T-1}(e)$, with $(st)^{O(\sqrt{\log(t)})}$ many queries to $(\cS_{i,0},\cE_{i,0})$ and the original set cover instance. Remember that $\Oracle_{i,T-1}(S) = \Oracle_i(S)$ and $\OracleEl_{i,T-1}(e) = \OracleEl_i(e)$. This implies that we can answer $\Oracle_{i}(S)$ and $\OracleEl_{i}(e)$ with $(st)^{O(\sqrt{\log(t)})}$ many queries to $(\cS_{i,0},\cE_{i,0})$ and the original set cover instance.
	
\paragraph{Description of $\Oracle_{i}$ and $\OracleEl_{i}$ given query access to $\Oracle_{i-1}$ and $\OracleEl_{i-1}$ ($i > 1$).}
		It is easy to verify that we can give query access to $(\cS_{i,0},\cE_{i, 0})$ by using $O(st)$ many queries to  $\Oracle_{i-1}$, $\OracleEl_{i-1}$ and the original set cover instance. Together with the query complexity analysis of the previous paragraph, this implies that we can answer $\Oracle_{i}(S)$ and $\OracleEl_{i}(e)$ with $(st)^{O(\sqrt{\log(t)})}$ many oracle calls to $\Oracle_{i-1}$ and $\OracleEl_{i-1}$ and the original set cover instance.
	
\paragraph{Description of $\Oracle$ given query access to $(\cS,\cE)$.}
	Note that $(\cS_{1,0},\cE_{1,0}) = (\cS,\cE)$. Hence, a simple induction argument implies that we can answer $\Oracle_{i}(S)$ and $\OracleEl_{i}(e)$ with $(st)^{O(i \cdot \sqrt{\log(t)})}$ many queries to $(\cS,\cE)$. In particular, this implies that we can answer $\Oracle_{\log(s)}(S)$ with $(st)^{O(\log(s) \cdot \sqrt{\log(t)})}$ many queries to $(\cS,\cE)$. Now, it is straightforward to implement $\Oracle(S)$ with $O(st)$ many calls to $\Oracle_{\log(s)}$ and $(\cS,\cE)$. This implies that the total query complexity is $(st)^{O(\log(s) \cdot \sqrt{\log(t)})}$.

\subsection{Approximation Proof}
In this section we analyze the approximation guarantee of \cref{alg:sqrt-t} and prove the following result.
\begin{theorem}\label{thm:main-theorem-of-sqrt-t-approximation}
Let $\cS_{\cvr}$ be the solution returned by \cref{alg:sqrt-t}. Then, $\ee{|\cS_{\cvr}|} = O(\log s) \cdot \opt$.
\end{theorem}
To prove \cref{thm:main-theorem-of-sqrt-t-approximation}, we highly reuse the statements proved for \cref{alg:generic-alg}. In particular, \cref{theorem:gen-set-cover-size} shows that \cref{alg:generic-alg} in expectation returns a $O(\log{s})$-approximate solution. In this section, we will analyze the similarity between \cref{alg:sqrt-t} and \cref{alg:generic-alg}, and show that from the point of view of most of the elements and sets these two algorithms are identical. Then we apply \cref{theorem:gen-set-cover-size} to prove \cref{thm:main-theorem-of-sqrt-t-approximation}.

We now briefly describe the notation and the setup that is used for proving the required lemmas. 
First, we randomly choose sets $B_i(S)$ and families $\cS_{i,k}$ as described in \cref{line:gen-define-Bi-and-pi} and \cref{line:gen-define-Sik} of \cref{alg:generic-alg}. After that we execute \cref{alg:generic-alg} and \cref{alg:sqrt-t} ``simultaneously'' with the same $B_i(S)$'s and  $\cS_{i,k}$'s. For $r \in \{0,1,\ldots ,\log s\cdot \log t\}$, some set $S$ is considered to have the same state after $r$ rounds in both algorithms iff $S$ was either added to the set cover in both algorithms or in none of them after $r$ rounds. Similarly, an element $e$ is considered to be in the same state in both algorithms after $r$ rounds iff $e$ is still free after $r$ rounds in both algorithms or $e$ is covered in both algorithms after $r$ rounds. We say that the $T$-hop neighborhood of some set $S$ (resp., element $e$) is in the same state after $r$ rounds in both algorithms iff all the sets and all the elements in the $T$-hop neighborhood of $S$ (resp., $e$) have the same state in both algorithms after $r$ rounds.

\begin{lemma} \label{lemma:sqrt-t-ten-hop}
Consider an arbitrary round $r < \log s \cdot \log t$. For a set $S$ (an element $e$, respectively) let $EQ^{(r)}_{S,10}$ ($EQ^{(r)}_{e,10}$, respectively) denote the event that the $10$-hop neighborhood of $S$ (e, respectively) in \cref{alg:generic-alg} and \cref{alg:sqrt-t} is in the same state after $r$ rounds. Furthermore, let $DIF^{(r+1)}_S$ ($DIF^{(r+1)}_e$, respectively) denote the event that $S$ ($e$, respectively) is in a different state in \cref{alg:generic-alg} compared to \cref{alg:sqrt-t} after $r+1$ rounds. Assume that $\prob{EQ^{(r)}_{S,10}} \geq \frac{1}{2}$ and $\prob{EQ^{(r)}_{e,10}} \geq \frac{1}{2}$. Then,
\[
	\prob{DIF^{(r+1)}_S\ |\ EQ^{(r)}_{S,10}} = \negthree \text{, and }
\]
\[
	\prob{DIF^{(r+1)}_e\ |\ EQ^{(r)}_{e,10}} = \negthree.
\]
\end{lemma}
\begin{proof} 
	We provide the proof for $S$, but the same argument works for $e$. Assume that the $10$-hop neighborhood of $S$ is in the same state after $r$ rounds in both algorithms. In order for $S$ to be in a different state after $r+1$ rounds, either
\begin{enumerate}[(A)]
	\item\label{item:case-S'-added-to-the-cover} there exists a set $S'$ in the $5$-hop neighborhood of $S$ such that $S'$ is added to the set cover in \cref{line:alg2_large_set} of \cref{alg:sqrt-t} during round $r+1$; or
	\item\label{item:case-e'-pretends-to-be-covered} there exists an element $e'$ in the $5$-hop neighborhood of $S$ such that $e'$ pretends to be covered in \cref{line:sqrt-t-pretend} of \cref{alg:sqrt-t} during round $r+1$.
\end{enumerate}
We consider these two cases independently.
	
\paragraph{Case~\ref{item:case-S'-added-to-the-cover}.} Let $i$ be the stage that $r$ is part of. This implies that $B_i(S')$ contains more than $\logten$ many free elements at the beginning of stage $i$ in \cref{alg:sqrt-t}. As the $5$-hop neighborhood of $S'$ is contained in the $10$-hop neighborhood of $S$, this implies that $B_i(S')$ also contains more than $\logten$ many elements in \cref{alg:generic-alg}. This implies that $S'$ is a bad set in round $r$.
	
\paragraph{Case~\ref{item:case-e'-pretends-to-be-covered}.} Observe that an element can start pretending to be covered from round $r + 1$ onwards only if $r + 1$ corresponds to the first iteration of some phase $(i,j)$. Hence, assume that $r + 1$ corresponds to the first iteration of phase $j$ in stage $i$ of \cref{alg:sqrt-t} and to iteration $k_1 = j \cdot T + 1$ in stage $i$ of \cref{alg:generic-alg}. As $e'$ pretends to be covered, $e'$ was still uncovered at the beginning of iteration $k_1$. Therefore, there exists $k_2 \in \{j \cdot T + 1, \ldots,(j+1) \cdot T\}$ such that $\cS_{i,k_2}$ had more than $\logten \cdot 2^T \geq \logten \cdot 2^{k_2-k_1}$ many sets that contain $e'$ with an estimated set size of at least  $s/2^i$ at the beginning of iteration $k_1$ in stage $i$ (see \cref{line:sqrt-t-if-high-degree-element} of \cref{alg:sqrt-t}). This implies that $e'$ is a bad element in \cref{alg:generic-alg}. 
	
\paragraph{Combining the two cases.}
	Both cases imply that $S$ has a bad neighborhood in \cref{alg:generic-alg}. Therefore, from \cref{lemma:gen-bad-neighborhood} we conclude
	\begin{align*}
		\prob{DIF^{(r+1)}_S\ |\ EQ^{(r)}_{S,10}} 
		&\leq \prob{\text{``$S$ has a bad neighborhood''}\ |\ EQ^{(r)}_{S,10}}  \\
		&\leq \frac{\prob{\text{``$S$ has a bad neighborhood''}}}{\prob{EQ^{(r)}_{S,10}}} \\
		&\leq 2 \cdot \prob{\text{``$S$ has a bad neighborhood''}} \\
		&\leq 2 \cdot \negthree.
	\end{align*}
In the same way we derive	$\prob{DIF^{(r+1)}_e\ |\ EQ^{(r)}_{e,10}} = \negthree$.
\end{proof}

\begin{lemma}\label{lemma:p_r-is-small}
	Consider a round $r \in \{0,\ldots,\log s \cdot \log t\}$. Define $T_r := 10 \cdot (\log s \cdot \log t + 1 - r)$. Let $p_r$ be equal to the maximum probability, over all sets $S$ and all elements $e$, that the $T_r$-hop neighborhood of $S$ or $e$ looks different in \cref{alg:generic-alg} compared to \cref{alg:sqrt-t} after $r$ rounds. Then, it holds
	\begin{equation}\label{eq:upper-bound-on-p_r}
		p_r \leq r \cdot c \cdot \negonenoO,
	\end{equation}
	for an absolute constant $c$.
\end{lemma}
\begin{proof}
	We proof the lemma by induction over $r$. The base case $r = 0$ holds trivially. Let $r \in \{0, \ldots,\log s \cdot \log t - 1\}$ be an arbitrary round. Assuming that \cref{eq:upper-bound-on-p_r} holds for $r$, we will show that \cref{eq:upper-bound-on-p_r} holds for $r + 1$ as well.
	
	Let $S$ be a set. Assume that the  $T_{r+1}$-hop neighborhood of $S$ looks different in \cref{alg:generic-alg} compared to \cref{alg:sqrt-t} after $r+1$ rounds. This could happen for one of the two following reasons. First, the $T_r$-hop neighborhood of $S$ looked different after $r$ rounds in \cref{alg:generic-alg} compared to \cref{alg:sqrt-t}, which happens with a probability of at most $p_r$. Second, there exists a set $S'$ (an element $e'$, respectively) in the $T_{r+1}$-hop neighborhood of $S$ such that the $10$-hop neighborhood of $S'$ ($e'$, respectively) looks the same after $r$ rounds, but $S'$ ($e'$, respectively) is not in the same state after $r+1$ rounds in the two different algorithms. Conditioned that the first case does not happen, \cref{lemma:sqrt-t-ten-hop} together with a union bound over at most $(st)^{O(\log s \cdot \log t)}$ many sets and elements in the $T_{r+1}$-hop neighborhood of $S$ implies that this case happens with a probability of at most $c \cdot \negonenoO$, where $c$ is an absolute constant. This implies that
	\begin{align*}
		p_{r+1} \leq p_r + c \cdot \negonenoO
		\leq r \cdot c \cdot \negonenoO + c \cdot \negonenoO 
		= (r+1) \cdot c \cdot \negonenoO,
	\end{align*}
as desired.
\end{proof}
We are now ready to prove the main theorem of this section.
\begin{proof}[Proof of \cref{thm:main-theorem-of-sqrt-t-approximation}]
	Let $p_S$ ($p'_S$, respectively) be the probability that $S$ is added to the set cover constructed by \cref{alg:generic-alg} (\cref{alg:sqrt-t}, respectively). \cref{lemma:p_r-is-small} implies that the $T_{\log(s) \cdot \log (t)}$-hop neighborhood of $S$ is in a different state in \cref{alg:sqrt-t} compared to \cref{alg:generic-alg} with a probability at most $p_{\log s \log t} \leq \frac{1}{s \cdot t}$. In particular, this implies that $|p_S - p'_S| \leq \frac{1}{s \cdot t}$. Together with \cref{theorem:gen-set-cover-size}, this implies that the expected size of $\cS_{\cvr}$ is at most $O(\log s) \cdot \opt + \frac{n \cdot t}{t \cdot s}$. This implies $\ee{|\cS_{\cvr}|} = O(\log s) \cdot \opt$ as desired.
\end{proof}

\newcommand{\numsetsthree}{\frac{\log^2 s \cdot \log^3 t \cdot 2^{R}}{R}}
\newcommand{\numsetsthreehalf}{\frac{\log^2 s \cdot \log^3 t \cdot 2^{R/2}}{R/2}}

\newcommand{\threshhold}{\log^5 s \log^5 t \cdot \frac{2^{k-k'}}{2^{R/2}}}

\newcommand{\sss}{\hcS_{[k',k'+R-1]}}
\newcommand{\sssplushalf}{\hcS_{[k'+R/2,k'+R-1]}}
\newcommand{\silp}{\cS'_{i,l}}

\section{Repeated Sparsification}\label{sec:log-log-t}
We now discuss how to reduce the query complexity from $(st)^{O(\log{s} \sqrt{\log{t}})}$ to  $s^{O(\log{s})} t^{O(\log{s} \cdot (\log \log{s} + \log \log{t}))}$, and hence prove \cref{thm:main-LCA}. We will first present the high-level ideas of our approach and convey the main intuition behind our algorithm, that we present as \cref{alg:main-alg,alg:rec-split,alg:base_case}.

Estimating set sizes will be still done in the same way as in \cref{alg:generic-alg}. Thus, our main focus lies on improving the sparsification of the element side. For the sake of cleaner presentation, we first review how sparsification works in \cref{alg:sqrt-t}. Let $i$ be a stage. Consider the moment when \cref{alg:sqrt-t} tests whether an element $e$ is covered or not in the iteration $\log{t}$ of stage $i$. First, note that the family $\cS_{i,\log t}$ contains all the sets in the set cover instance. Hence, $e$ might be contained in up to $t$ sets in $\cS_{i,\log t}$. However, recall that we sparsify $\cS_{i,\log t}$ at the beginning of the last phase by removing all the sets which have a small estimated number of free elements. The resulting family of sets is called $\hcS_{i,\log t}$ and $e$ only needs to check for all sets in $\hcS_{i,\log t}$ containing $e$ whether they get added to cover in iteration $\log t$ of stage $i$.

However, to obtain the relevant sets in $\hcS_{i,\log t}$, we need to estimate the degree of $t$ sets at the beginning of the last phase. This can be query-wise wasteful, in a sense that invoking oracles that answer queries about the last iteration/phase is significantly more expensive than invoking oracles which answer queries about iterations that happen early in the computation. Motivated by this observation, our goal now is to sparsify $\cS_{i,\log t}$ rather ``earlier'' than ``later''. We achieve that in multiple steps. More concretely, we first estimate the degree of all the sets that contain $e$ in $\cS_{i,\log t}$ after $\log (t)/2$ iterations. Given that $e$ is still uncovered after iteration $\log(t)/2$, we expect no more than $\sqrt{t}$ sets to have a large degree. For all of the at most $\sqrt{t}$ sets, we now estimate the degree after $\frac{3}{4}\log t$ iterations. If $e$ is still uncovered after iteration $\frac{3}{4}\log t$, we expect no more than $t^{1/4}$ of those sets to have a large degree. This process is repeated until we expect no more than $\poly(\log t)$ of the sets to have a large degree. At that point, $e$ is contained in a small number of sets that have to be tested, and we simply test for each of those sets whether they get added to the set cover during iteration $\log(t)$ of stage $i$ or not.

Now we describe the structure of \cref{alg:main-alg}. In order to execute stage $i$, \cref{alg:main-alg} calls the recursive procedure \cref{alg:rec-split}. The top-level call of \cref{alg:rec-split} executes $\log(t)$ iterations. As the input, \cref{alg:rec-split} receives the families of sets $\hcS_{1}, \ldots, \hcS_{\log t}$, which contain all the sets that will potentially be added to the set cover in iterations 1 to $\log t$. Given that, \cref{alg:rec-split} calls itself recursively to execute the first $\log(t)/2$ iterations. It passes down the families of sets  $\hcS_{1}, \ldots, \hcS_{\log(t)/2}$ that contain all the sets that potentially get added to the set cover between iterations 1 and $\log (t)/2$. Note that we expect $e$ to be contained in no more than $\sqrt{t}$ sets of $\hcS_{\log(t)/2}$. After executing the first $\log(t)/2$ iterations, we recursively call \cref{alg:rec-split} to execute the remaining $\log(t)/2$ iterations. However, we do not simply pass down the recursion families of sets $\hcS_{\log(t)/2+1}, \ldots, \hcS_{\log(t)}$. Instead, we first sparsify the families of sets by removing all the sets with an estimated degree less than $s/2^i$. Then, the resulting families of sets $\cS'_{\log(t)/2+1}, \ldots, \cS'_{\log(t)}$ get passed down to the second recursive call. Furthermore, for a free element $e$, we do not expect $e$ to be contained in more than $\sqrt{t}$ sets in any of the families of sets that get passed to the second recursive call. If we detect any such element $e'$, as before, we mark $e'$ as pretending to be covered.

\begin{algorithm}
	\caption{\label{alg:main-alg}A variant of \cref{alg:generic-alg} that can be simulated by an LCA with a query complexity of $s^{O(\log{s})} t^{O(\log{s} \cdot (\log \log{s} + \log \log{t}))}$.}
	\For{each pair of $(i,k)$ where $i \in [\log s]$ and $k \in [\log t]$}{		
		$B_i(S)$ and $\cS_{i,k}$ are generated in the exact same way as in \cref{alg:generic-alg} 		
	}
	\For{stage i = 1 to $\log s$}{	
			\For{each set $S \in \cS$}{
				$\hat{B}_{i}(S) \gets B_i(S) \cap \cE_{i}$\label{line:hat-B} 
				\tcp{$\cE_{i}$ denotes the set of free elements in the beginning of the stage}
				\If(\tcp*[h]{test whether $S$ is a bad set}){$|\hat{B}_{i}(S)| \geq \logten$}{
					{\bf add} $S$ to $\cS_{\cvr}$.  \label{line:alg3_large_set}
				}
			}

	\For{each $k \in [\log (t)]$}{		 
		
		$\hcS_{i,k} \gets \{S|S \in S_{i,k}, \degree(S) \geq s/2^i\}$ 
		
				\For{each $e\in \cE_i$ in parallel\label{line:alg3-discard-high-degree-els}}{
					\If(\tcp*[h]{test whether $e$ is a bad element})
					{$d_{\hsil}(e) \geq \logten \cdot 2^T$\label{line:alg3-if-high-degree-element}}{				   			
						{\it pretend} that $e$ is covered. 
				   		\tcp{$e$ will be covered in~\cref{line:alg3-clean-up}.}\label{line:alg3-pretend}
				   	}
				}
}

		$\hcE_i \gets $ all elements in $\cE_i$ which are uncovered (and don't pretend to be covered)
		$\RecSplit(1, \log t, i, (\hcS_{i,1},\hcS_{i,2},...,{\hcS_{i,\log t}}), \hcE_i)$
	}
	\For(\tcp*[h]{Handling bad elements}){each free element $e$}{
		{\bf add} the set with smallest ID which contains $e$ to $\cS_{\cvr}$.\label{line:alg3-clean-up}
	}
\end{algorithm}

\begin{algorithm}
	\caption{$\RecSplit(k', R, i, \hcS_{[k',k'+R-1]} = (\hcS_{k'}, \hcS_{k'+1}, \ldots,\hcS_{k'+R - 1}), \hcE_i)$\label{alg:rec-split}}
  \KwIn{\\
		\quad Simulates iterations $k'$ to $k'+R-1$ of the $i$-th iteration  \\
		\quad For $k \in \{k',...,k'+R-1\}$, $\hcS_k$ consists of sets that will be added to the set cover as long as the estimated number of free elements is at least $s/2^i$ in the beginning of the $k$-th iteration of stage $i$. \\
		\quad $\hat{\cE}_i$ consists of all elements which are uncovered (and don't pretend to be covered) at the beginning of iteration $k'$
  } 
	
	\If{$R \le \log \log{t}$}{
		\textsc{BaseCase}($k', R, i, (\hcS_{k'},\hcS_{k'+1}, \ldots,\hcS_{k'+R-1}), \hcE_i$) \\
		\Return
	}

	$\RecSplit(k', R/2, i, (\hcS_{k'},\hcS_{k'+1}, \ldots,\hcS_{k'+R/2 - 1}), \hcE_i)$ \label{line:1st-rec-call}

	\For{each $k \in \{k'+R/2, \ldots,k'+R-1\}$}{		 
		
		$\cS'_{k} \gets \{S|S \in \hcS_{k}, \degree(S) \geq s/2^i\}$ 
		
		\For{each uncovered element $e$ in parallel}{
		\If(\tcp*[h]{only happens if $e$ is a bad element}){$d_{\cS'_k}(e) \geq \logten \cdot 2^{k-(k' + R/2)}$ \label{line:loglog-if-high-degree-element2}}{

			{\it pretend} that $e$ is covered. 			\label{line:alg-3-preten-2}
		}
	}
}	
	
	$\cE'_i \gets $ all uncovered elements in $\hat{\cE}_i$ 
	
	$\RecSplit(k'+R/2, R/2, i, (\cS'_{k'+R/2},\cS'_{k'+R/2+1}, \ldots,\cS'_{k'+R-1}), \cE'_i)$ \label{line:2nd-rec-call}
\end{algorithm}

\begin{algorithm}
	\caption{\textsc{BaseCase}($k', R, i, \hcS_{[k',k'+R-1]} = (\hcS_{k'},\hcS_{k'+1}, \ldots,\hcS_{k'+ R - 1}), \hcE_i)$\label{alg:base_case}}

		\For{iteration $k = k'$ to $k' + R - 1$}{
			\For{all sets S in parallel}{
				\If(\tcp*[h]{$B_i$ is computed in \cref{line:hat-B}}){$S \in \hcS_k$ is not in the cover and $\degree(S) \geq s/2^i$}{ 
					Add $S$ to the set cover
				}
			}
		}		
\end{algorithm}

\subsection{Establishing an Invariant}

\begin{invariant}\label{invariant:RecSplit}
	We define the following invariant for $\RecSplit(k', R, i, \sss, \hcE_i)$:
	\begin{enumerate}[(1)]
		\item\label{item:invariant-element-contained-in} For $l \in \{0,...,R-1\}$, each element $e$ in $\hcE_i$ is contained in at most $\logten  \cdot 2^{l+1}$ many sets of $\hcS_{k'+ l}$.
	\end{enumerate}
\end{invariant}
This invariant directly follows from \cref{line:alg3-pretend} of \cref{alg:main-alg} and \cref{line:alg-3-preten-2} of \cref{alg:rec-split} together with a simple induction argument. In particular, \cref{invariant:RecSplit} implies that each uncovered element is contained in at most $\poly(\log s) \cdot 2^{O(R)}$ many sets in $\hcS_k$ for any $k \in \{k',...,k'+R-1\}$.

\subsection{LCA Simulation of \cref{alg:rec-split} and \cref{alg:base_case}}
In this section we describe the LCA implementation of \cref{alg:rec-split} and \cref{alg:base_case}. To simplify notation, we sometimes treat $\sss$ as a set consisting of all sets $S$ which are contained in $\hcS_k$ for $k \in \{k',...,k'+R-1\}$.

For the sake of analysis, we define two oracles $\Oracle_{k',R,i}(S)$ and $\OracleEl_{k',R,i}(e)$. $\Oracle_{k',R,i}$ takes as an input a set $S$ and returns whether $S$ gets added to the set cover during the execution of $\RecSplit(k',R,i,\sss,\hcE_i)$. Similarly, $\OracleEl_{k',R,i}$ takes as an input an arbitrary element $e$ and returns whether $e$ gets covered for the first time (or pretends to be covered) during the execution of $\RecSplit(k',R,i,\sss,\hcE_i)$. In the following two paragraphs, we show how to implement the two oracles by having query access to the following type of queries for some arbitrary element $e$ or some arbitrary set $S$:
\begin{enumerate}[(1)]
	\item For some $k \in \{k', \ldots,k'+R-1\}$, is $S$ contained in $\hcS_k$?
	\item For some $k \in \{k', \ldots,k'+R-1\}$ and some $e \in \hcE_i$, give me all sets in $\hcS_k$ which contain $e$.
	\item Is $e$ contained in $\hcE_i$?
	\item Give me all elements contained in $\hat{B}_i(S)$.
\end{enumerate}

We then derive a bound for the number of query accesses the algorithm needs.

\paragraph{Base Case: Description of  $\Oracle_{k',R,i}(S)$ and $\OracleEl_{k',R,i}(e)$ ($R \leq \log \log t$)}

As $R \leq \log \log t$, we only need to focus on \cref{alg:base_case}. In the following, we assume that $S \in \sss$ and $e \in \hcE_i$, as these two cases are easy to check and $S \notin \sss$ implies that $S$ won't be added to the set cover and $e \notin \hcE_i$ implies  that $e$ was already previously covered (or pretended to be covered). Now, consider the directed bipartite graph $G = (\sss \cup \hcE_i,E_1 \cup E_2)$ with $(S,e) \in E_1$ iff $e \in \hat{B}_i(S)$ and $(e,S) \in E_2$ iff $e \in \sss$. Note that in order to decide for some set $S \in \sss$ if it gets added to the set cover, or some element $e \in \hcE_i$ whether $e$ gets covered, it is sufficient to gather all the information within the $O(\log\log t)$-hop neighborhood of $S$ or $e$, respectively. Note that the maximal out-degree of some set in $G$ is $\logten$ and \cref{invariant:RecSplit} ensures that the maximal out-degree of an element is $\poly(\log s) \cdot 2^{O(R)}$. Thus, the maximal out-degree of $G$ is $O(\poly(\log s \cdot \log t))$. We can therefore gather all the relevant information with $O(\log\log t) \cdot (\log s \cdot \log t)^{O(\log\log t)}$ many queries. This implies that we can answer $\Oracle_{k',R,i}(S)$ and $\OracleEl_{k',R,i}(e)$ with $(\log s \cdot \log t)^{O(\log\log t)}$ many queries.
\paragraph{Recursive Case: Description of  $\Oracle_{k',R,i}(S)$ and $\OracleEl_{k',R,i}(e)$ ($R > \log \log t$)}

Let $Q(R)$ denote the query complexity to answer $\Oracle_{k',R,i}(S)$ and $\OracleEl_{k',R,i}(e)$, respectively. For $R > \log\log t$, we will establish that 

 $$Q(R) \leq \poly(\log S) \cdot 2^{O(R)} \cdot Q(R/2)^2$$

For some arbitrary set $S$, we first check if $S$ gets added to the set cover during the first recursive call by invoking $\Oracle_{k',R/2,i}(S)$. Note that $\Oracle_{k',R/2,i}(S)$ assumes query access to 4 different types of queries. However, those queries are basically the same queries (slightly more restricted) as the ones we assume to have to implement oracles $\Oracle_{k',R,i}$ and $\OracleEl_{k',R,i}$. Hence, we can answer $\Oracle_{k',R/2,i}(S)$ with $Q(R/2)$ many queries. Similarly, we can decide with $Q(R/2)$ many queries if $e$ gets covered during the first recursive call.

Note that the oracles corresponding to the second recursive call assume to have access to the following queries:
\begin{enumerate}[(1)]
	\item For some $k \in \{k'+R/2,...,k'+R-1\}$, is $S$ contained in $\cS'_k$?
	\item For some $k \in \{k'+R/2,...,k'+R-1\}$ and some $e \in \cE'_i$, give me all sets in $\cS'_k$ which contain $e$.
	\item Is $e$ contained in $\cE'_i$?
	\item Give me all elements contained in $\hat{B}_i(S)$.
\end{enumerate}
To decide if $S$ is contained in $S'_k$, we first check if $S$ is contained in $\hcS_k$. If yes, then we check for each element $e$ in $\hat{B}_i(S)$ if $e$ is still uncovered after the first recursive call. Given this information, we can decide if $S$ is contained in $S'_k$. Hence, it takes $O(\poly(\log s)\poly(\log t)) \cdot Q(R/2)$ many queries to decide whether $S$ is contained in $S'_k$. \\
In order to return all the sets in $S'_k$ which contain $e$ (We only assume $e \in \hcE_i$ instead of $e \in \cE'_i$), we first get all of the at most $\poly(\log s)\cdot 2^{O(R)}$ many sets in $\hcS_k$ which contain $e$. This takes one query. Then, we check for each of them whether it is also contained in $S'_k$. Hence, we can answer this query with $\poly(\log s) \cdot 2^{O(R)}\cdot Q(R/2)$ many queries. \\
To decide if $e \in \cE'_i$, we first check if $e \in \hcE_i$. If yes, then we check if $e$ gets covered during the first recursive call. If not, then we check if $e$ pretends to be covered by checking for each of the at most $\poly(\log s) \cdot 2^{O(R)}$ many sets in $\sss$ which contain $e$ whether they get added to one or more of the $S'_k$. Thus, answering if $e \in \cE'_i$ takes $\poly(\log s) \cdot 2^{O(R)} \cdot O(\poly(\log s)\poly (\log t)) \cdot Q(R/2)$ many queries. \\
The last query does not change across different recursive calls. Hence, we can answer each of those 4 queries with at most $\poly(\log s) \cdot 2^{O(R)}$ many queries. Hence, we need $\poly(\log s) \cdot 2^{O(R)} \cdot Q(R/2)^2$ many queries to decide if a set is added to the set cover during the second recursive call or some element gets covered during the second recursive call. Thus, we can answer the oracles $\Oracle_{k',R,i}$ and $\OracleEl_{k',R,i}$ with $\poly(\log s) \cdot 2^{O(R)} \cdot Q(R/2)^2$ many queries. 

\subsubsection{Bounding the Recursion}
By the previous two subsections, there exist two constants $c_1$ and $c_2$ such that we can upper bound the query complexity as follows

\begin{equation}
Q(R) \leq\begin{cases}
(\log s \log t)^{c_1 \cdot \log\log t}, &   R \leq \log\log t\\
 (\log s \cdot 2^R)^{c_2} Q(R/2)^2, & R > \log \log t
\end{cases}
\end{equation}

We will now show by induction that there exists some constant $c$ such that for  $R \geq \frac{1}{2} \cdot \log \log t$, we have: 

$$ Q(R) \leq \rb{(\log s) \cdot 2^R}^{c \cdot (R - 1)}$$

Note that this implies $Q(\log t) \leq (\log s)^{O(\log t)} \cdot t^{\cO(\log\log t)} = t^{O(\log\log t + \log \log s)}$. 
For $c$ large enough, it is easy to see that it holds for $\frac{1}{2}\log \log t \leq R \leq \log \log t$. Now, consider some $R > \log\log(t)$. We get

\begin{align*}
	Q(R) &\leq  (\log s \cdot 2^R)^{c_2} \cdot \rb{Q(R/2)}^2 \\
	&\leq (\log s \cdot 2^R)^{c_2} \cdot \rb{\rb{\log s \cdot 2^{R/2}}^{c \cdot (R/2-1)}}^2 \\
	&\leq (\log s \cdot 2^R)^{c_2} \cdot \rb{\log s \cdot 2^{R/2}}^{c \cdot R - 2c}\\
	&\leq \rb{\log s \cdot 2^R}^{c \cdot (R-1)}
\end{align*}

as desired.

\subsection{LCA implementation and query complexity of \cref{alg:main-alg}}
For some stage $i$, we need to answer $Q(\log(t))$ many of the following queries: 

\begin{enumerate}[(1)]
	\item For some $k \in \{1,..., \log t\}$, is $S$ contained in $\hcS_{i,k}$?
	\item For some $k \in \{1,..., \log t\}$ and some $e \in \cE_i$, give me all sets in $\hcS_{i,k}$ which contain $e$.
	\item Is $e$ contained in $\hcE_i$?
	\item Give me all elements contained in $\hat{B}_i(S)$.
\end{enumerate}

Each one of those can be answered with $O(st)$ many queries to $(\cS_i,\cE_i)$ and the original set cover instance. Given that, we can ``glue'' different stages together in the same way as in \cref{alg:sqrt-t}. Thus, the query complexity of \cref{alg:main-alg} is $(\poly(st)\cdot Q(\log t))^{\log s}$. Plugging in $t^{O(\log t\log t + \log \log s)}$ for $Q(\log t)$ leads to a query complexity of $t^{O(\log s (\log\log s + \log\log t))} \cdot s^{O(\log s)} = s^{O(\log t (\log\log s + \log\log t) + \log s)}$.

\subsection{Approximation Proof}
\begin{theorem}
	Let $\cS_{\cvr}$ be the solution returned by $\cref{alg:main-alg}$. Then, $\ee{|\cS_{\cvr}|} = O(\log s) \cdot \OPT$.
\end{theorem}
\begin{proof}
	The proof is completely analogous to the proof of \cref{thm:main-theorem-of-sqrt-t-approximation}. The only thing to show is the following: For some set $S$ (element $e$), given that the $10$-hop neighborhood of $S$ looks the same after $r$ rounds in both algorithms, $S$ can only be in a different state after round $r+1$ in \cref{alg:generic-alg} compared to \cref{alg:main-alg} if $S$ has a bad neighborhood in \cref{alg:generic-alg}. This can be seen as follows: $S$ can only be in a different state after $r+1$ rounds if either some element $e'$ in the $5$-hop neighborhood of $S$ pretends to be covered during round $r+1$ or some set $S'$ in the $5$-hop neighborhood of $S$ gets added to the set cover in \cref{line:alg3_large_set} of \cref{alg:main-alg}. This would however imply that $S$ has a bad neighborhood in \cref{alg:generic-alg} as desired.
\end{proof}

\section*{Acknowledgment}
The authors would like to thank Mohsen Ghaffari for his helpful comments in the various stages of this project. 	
S.~Mitrovi{\' c} was supported by the Swiss NSF grant P2ELP2\_181772, MIT-IBM Watson AI Lab and Research Collaboration Agreement No.~W1771646, and FinTech@CSAIL. A.~Vakilian was supported by the NSF grant CCF-1535851. R.~Rubinfeld was supported by the MIT-IBM Watson AI Lab and Research Collaboration Agreement No. W1771646, the NSF awards CCF-1740751, CCF-1733808, and IIS-1741137, and FinTech@CSAIL.

\bibliographystyle{abbrv}
\bibliography{ref-lca-sc}

\begin{thebibliography}{10}

\bibitem{ams-acskr-06}
N.~Alon, D.~Moshkovitz, and S.~Safra.
\newblock Algorithmic construction of sets for $k$-restrictions.
\newblock {\em ACM Trans. Algo.}, 2(2):153--177, 2006.

\bibitem{alon2012space}
N.~Alon, R.~Rubinfeld, S.~Vardi, and N.~Xie.
\newblock Space-efficient local computation algorithms.
\newblock In {\em Proceedings of the twenty-third annual ACM-SIAM symposium on
  Discrete Algorithms}, pages 1132--1139. Society for Industrial and Applied
  Mathematics, 2012.

\bibitem{a-tsatmpsscp-17}
S.~Assadi.
\newblock {Tight Space-Approximation Tradeoff for the Multi-Pass Streaming Set
  Cover Problem}.
\newblock In {\em Proc. 36th ACM Sympos. on Principles of Database Systems
  {\em(PODS)}}, pages 321--335, 2017.

\bibitem{aky-tbspscscp-16}
S.~Assadi, S.~Khanna, and Y.~Li.
\newblock Tight bounds for single-pass streaming complexity of the set cover
  problem.
\newblock In {\em Proc. 48th Annu. ACM Sympos. Theory Comput. {\em(STOC)}},
  pages 698--711, 2016.

\bibitem{bansal2009new}
N.~Bansal, A.~Caprara, and M.~Sviridenko.
\newblock A new approximation method for set covering problems, with
  applications to multidimensional bin packing.
\newblock {\em SIAM Journal on Computing}, 39(4):1256--1278, 2009.

\bibitem{bem-aosacp-17}
M.~Bateni, H.~Esfandiari, and V.~S. Mirrokni.
\newblock Almost optimal streaming algorithms for coverage problems.
\newblock In {\em Proc. 29th ACM Sympos. Parallel Alg. Arch. {\em(SPAA)}},
  pages 13--23, 2017.

\bibitem{berger1994efficient}
B.~Berger, J.~Rompel, and P.~W. Shor.
\newblock Efficient {NC} algorithms for set cover with applications to learning
  and geometry.
\newblock {\em Journal of Computer and System Sciences}, 49(3):454--477, 1994.

\bibitem{blelloch2011linear}
G.~E. Blelloch, R.~Peng, and K.~Tangwongsan.
\newblock Linear-work greedy parallel approximate set cover and variants.
\newblock In {\em Proceedings of the twenty-third annual ACM symposium on
  Parallelism in algorithms and architectures}, pages 23--32. ACM, 2011.

\bibitem{blelloch2012parallel}
G.~E. Blelloch, H.~V. Simhadri, and K.~Tangwongsan.
\newblock Parallel and i/o efficient set covering algorithms.
\newblock In {\em Proceedings of the twenty-fourth annual ACM symposium on
  Parallelism in algorithms and architectures}, pages 82--90. ACM, 2012.

\bibitem{cw-igpcs-16}
A.~Chakrabarti and A.~Wirth.
\newblock Incidence geometries and the pass complexity of semi-streaming set
  cover.
\newblock In {\em Proc. 27th ACM-SIAM Sympos. Discrete Algs. {\em(SODA)}},
  pages 1365--1373, 2016.

\bibitem{chang2019complexity}
Y.-J. Chang, M.~Fischer, M.~Ghaffari, J.~Uitto, and Y.~Zheng.
\newblock The complexity of ($\delta$+ 1) coloring in congested clique,
  massively parallel computation, and centralized local computation.
\newblock In {\em Proceedings of the 2019 ACM Symposium on Principles of
  Distributed Computing}, pages 471--480, 2019.

\bibitem{ckw-scavl-10}
G.~Cormode, H.~J. Karloff, and A.~Wirth.
\newblock Set cover algorithms for very large datasets.
\newblock In {\em Proc. 19th {ACM} Conf. Info. Know. Manag. {\em(CIKM)}}, pages
  479--488, 2010.

\bibitem{dimv-sccsc-14}
E.~D. Demaine, P.~Indyk, S.~Mahabadi, and A.~Vakilian.
\newblock On streaming and communication complexity of the set cover problem.
\newblock In {\em Proc. 28th Int. Symp. Dist. Comp. {\em(DISC)}}, volume 8784,
  pages 484--498, 2014.

\bibitem{ds-aapr-14}
I.~Dinur and D.~Steurer.
\newblock Analytical approach to parallel repetition.
\newblock In {\em Proc. 46th Annu. ACM Sympos. Theory Comput. {\em(STOC)}},
  pages 624--633, 2014.

\bibitem{er-sssc-14}
Y.~Emek and A.~Ros{\'{e}}n.
\newblock Semi-streaming set cover.
\newblock In {\em Proc. 41st Int. Colloq. Automata Lang. Prog. {\em(ICALP)}},
  volume 8572 of {\em Lect. Notes in Comp. Sci.}, pages 453--464, 2014.

\bibitem{feige1998threshold}
U.~Feige.
\newblock A threshold of ln n for approximating set cover.
\newblock {\em Journal of the ACM (JACM)}, 45(4):634--652, 1998.

\bibitem{ghaffari2019sparsifying}
M.~Ghaffari and J.~Uitto.
\newblock Sparsifying distributed algorithms with ramifications in massively
  parallel computation and centralized local computation.
\newblock In {\em Proceedings of the Thirtieth Annual ACM-SIAM Symposium on
  Discrete Algorithms (SODA)}, pages 1636--1653, 2019.

\bibitem{grigoriadis1995sublinear}
M.~D. Grigoriadis and L.~G. Khachiyan.
\newblock A sublinear-time randomized approximation algorithm for matrix games.
\newblock {\em Operations Research Letters}, 18(2):53--58, 1995.

\bibitem{gw-ceaas-97}
T.~Grossman and A.~Wool.
\newblock Computational experience with approximation algorithms for the set
  covering problem.
\newblock {\em Euro. J. Oper. Res.}, 101(1):81--92, 1997.

\bibitem{himv-ttbsscp-16}
S.~{Har-Peled}, P.~Indyk, S.~Mahabadi, and A.~Vakilian.
\newblock Towards tight bounds for the streaming set cover problem.
\newblock In {\em Proc. 35th ACM Sympos. on Principles of Database Systems
  {\em(PODS)}}, pages 371--383, 2016.

\bibitem{imruvy-fscsm-17}
P.~Indyk, S.~Mahabadi, R.~Rubinfeld, J.~Ullman, A.~Vakilian, and
  A.~Yodpinyanee.
\newblock Fractional set cover in the streaming model.
\newblock In {\em 20th International Workshop on Approximation Algorithms for
  Combinatorial Optimization Problem (APPROX 2017)}, pages 198--217, 2017.

\bibitem{IndykMRVY18}
P.~Indyk, S.~Mahabadi, R.~Rubinfeld, A.~Vakilian, and A.~Yodpinyanee.
\newblock Set cover in sub-linear time.
\newblock In {\em Proceedings of the Twenty-Ninth Annual ACM-SIAM Symposium on
  Discrete Algorithms (SODA)}, pages 2467--2486, 2018.

\bibitem{johnson1974approximation}
D.~S. Johnson.
\newblock Approximation algorithms for combinatorial problems.
\newblock {\em Journal of computer and system sciences}, 9(3):256--278, 1974.

\bibitem{kv-iclt-1994}
M.~J. Kearns and U.~V. Vazirani.
\newblock {\em An introduction to computational learning theory}.
\newblock MIT press, 1994.

\bibitem{Koufogiannakis2014}
C.~Koufogiannakis and N.~E. Young.
\newblock A nearly linear-time {PTAS} for explicit fractional packing and
  covering linear programs.
\newblock {\em Algorithmica}, 70(4):648--674, 2014.

\bibitem{kuhn2006price}
F.~Kuhn, T.~Moscibroda, and R.~Wattenhofer.
\newblock The price of being near-sighted.
\newblock In {\em Proceedings of the seventeenth annual ACM-SIAM symposium on
  Discrete algorithm}, pages 980--989. Society for Industrial and Applied
  Mathematics, 2006.

\bibitem{kmvv-fgams-13}
R.~Kumar, B.~Moseley, S.~Vassilvitskii, and A.~Vattani.
\newblock Fast greedy algorithms in {MapReduce} and streaming.
\newblock In {\em Proc. 25th ACM Sympos. Parallel Alg. Arch. {\em(SPAA)}},
  pages 1--10, 2013.

\bibitem{lovasz1975ratio}
L.~Lov{\'a}sz.
\newblock On the ratio of optimal integral and fractional covers.
\newblock {\em Discrete mathematics}, 13(4):383--390, 1975.

\bibitem{m-pgcnsc-12}
D.~Moshkovitz.
\newblock The projection games conjecture and the {NP}-hardness of {$\ln
  n$}-approximating set-cover.
\newblock In {\em Approximation, Randomization, and Combinatorial Optimization.
  Algorithms and Techniques}, pages 276--287. 2012.

\bibitem{nguyen2008constant}
H.~N. Nguyen and K.~Onak.
\newblock Constant-time approximation algorithms via local improvements.
\newblock In {\em 2008 49th Annual IEEE Symposium on Foundations of Computer
  Science}, pages 327--336. IEEE, 2008.

\bibitem{onak2018round}
K.~Onak.
\newblock Round compression for parallel graph algorithms in strongly sublinear
  space.
\newblock {\em arXiv preprint arXiv:1807.08745}, 2018.

\bibitem{parnas2007approximating}
M.~Parnas and D.~Ron.
\newblock Approximating the minimum vertex cover in sublinear time and a
  connection to distributed algorithms.
\newblock {\em Theoretical Computer Science}, 381(1-3):183--196, 2007.

\bibitem{rs-sbepl-97}
R.~Raz and S.~Safra.
\newblock A sub-constant error-probability low-degree test, and a sub-constant
  error-probability {PCP} characterization of {NP}.
\newblock In {\em Proc. 29th Annu. ACM Sympos. Theory Comput. {\em(STOC)}},
  1997.

\bibitem{rubinfeld2011fast}
R.~Rubinfeld, G.~Tamir, S.~Vardi, and N.~Xie.
\newblock Fast local computation algorithms.
\newblock {\em arXiv preprint arXiv:1104.1377}, 2011.

\bibitem{sg-mcsma-09}
B.~Saha and L.~Getoor.
\newblock On maximum coverage in the streaming model {\&} application to
  multi-topic blog-watch.
\newblock In {\em Proc. {SIAM} Int. Conf. Data Mining {\em(SDM)}}, pages
  697--708, 2009.

\bibitem{schmidt1995chernoff}
J.~P. Schmidt, A.~Siegel, and A.~Srinivasan.
\newblock {Chernoff--Hoeffding} bounds for applications with limited
  independence.
\newblock {\em SIAM Journal on Discrete Mathematics}, 8(2):223--250, 1995.

\bibitem{vadhan2012pseudorandomness}
S.~P. Vadhan et~al.
\newblock Pseudorandomness.
\newblock {\em Foundations and Trends{\textregistered} in Theoretical Computer
  Science}, 7(1--3):1--336, 2012.

\bibitem{yoshida2012improved}
Y.~Yoshida, M.~Yamamoto, and H.~Ito.
\newblock Improved constant-time approximation algorithms for maximum matchings
  and other optimization problems.
\newblock {\em SIAM Journal on Computing}, 41(4):1074--1093, 2012.

\end{thebibliography}

\end{document}